\documentclass[11pt]{article}
\usepackage{fullpage}
\usepackage{amsmath,amssymb,amsthm}
\usepackage{xcolor,graphicx,url}
\usepackage{appendix}

\usepackage[normalem]{ulem}
\usepackage[hidelinks]{hyperref}

\allowdisplaybreaks

\newtheorem{theorem}{Theorem}[section]

\newtheorem{lemma}[theorem]{Lemma}
\newtheorem{claim}[theorem]{Claim}
\newtheorem{proposition}[theorem]{Proposition}
\newtheorem{fact}[theorem]{Fact}
\newtheorem{observation}[theorem]{Observation}

\theoremstyle{definition}
\newtheorem{definition}[theorem]{Definition}

\theoremstyle{remark}

\usepackage[boxruled, noend, noline]{algorithm2e}
\newcommand{\n}{\ensuremath{n}} 
\newcommand{\eps}{\varepsilon}
\newcommand{\Dyes}{\mathcal{D}^+}
\newcommand{\Dno}{\mathcal{D}^-}
\newcommand{\Dhat}{\hat{\mathcal{D}}^-}
\newcommand{\hypercube}{\ensuremath{\{0,1\}^\n}}
\newcommand{\maj}{\operatorname{Maj}}
\newcommand{\E}{\operatorname{\mathbb{E}}}
\newcommand{\switch}{\operatorname{\mathbb{S}}}

\newcommand{\ord}[1]{\ensuremath{{#1}^\text{th}}}
\newcommand{\ceil}[1]{\ensuremath{\left\lceil {#1} \right\rceil}}
\newcommand{\floor}[1]{\ensuremath{\left\lfloor {#1} \right\rfloor}}
\newcommand{\viol}[1]{\ensuremath{D_f^-({#1})}}
\newcommand{\dviol}[1]{\ensuremath{A_f^-({#1})}}
\newcommand{\goodviol}[1]{\ensuremath{E_f^-({#1})}}

\newcommand{\EdgeViolations}{\texttt{Edge-Violations}}
\newcommand{\eEdgeViolations}{\emph{\EdgeViolations}}
\newcommand{\EventEstimation}{\texttt{Matching-Estimation}}
\newcommand{\eEventEstimation}{\emph{\EventEstimation}}
\newcommand{\Event}{\texttt{Capture}}
\newcommand{\eEvent}{\emph{\Event}}
\newcommand{\CorrectEstimates}{\texttt{CorrectEstimates}}
\newcommand{\appfactor}{\ensuremath{\eps /\sqrt{n\log n}}}
\newcommand{\appfactorc}{\ensuremath{c_0\cdot \appfactor}}
\newcommand{\af}{\ensuremath{\gamma(n,\eps)}}
\newcommand{\eventfar}{\texttt{FAR}}
\newcommand{\eventbad}{\texttt{BAD}}

\newcommand{\ApproxMono}{\texttt{ApproxMono}}
\newcommand{\eApproxMono}{\emph{\ApproxMono}}
\newcommand{\close}{\textsf{close}}
\newcommand{\far}{\textsf{far}}
\newcommand{\poly}{\mathrm{poly}}
\newcommand{\dist}{\mathrm{dist}}

\newcommand{\epsf}{\ensuremath{\eps_f}}
\newcommand{\Ex}{\mathop{\E}}
\newcommand{\Prx}{\mathop{\Pr}}
\newcommand{\bS}{\boldsymbol{S}}
\newcommand{\bx}{\boldsymbol{x}}
\newcommand{\bi}{\boldsymbol{i}}
\newcommand{\red}{\mathrm{red}}
\newcommand{\blue}{\mathrm{blue}}
\newcommand{\calS}{\mathcal{S}}

\newcommand{\cA}{\ensuremath{\mathcal{A}}}
\newcommand{\cB}{\ensuremath{\mathcal{B}}}
\newcommand{\cP}{\ensuremath{\mathcal{P}}}
\newcommand{\col}{\mathrm{col}}
\newcommand{\eee}{\ensuremath{\mathrm{e}}}

\newcommand{\bc}{counter}
\newcommand{\bw}{\boldsymbol{w}}
\newcommand{\by}{\boldsymbol{y}}
\newcommand{\bz}{\boldsymbol{z}}
\newcommand{\EE}{$\EventEstimation(S, \delta, f)$}
\newcommand{\EV}{$\EdgeViolations(\delta, f)$}
\newcommand{\EventArgs}{\Event(x,{\bS},f)}
\newcommand{\eEventArgs}{\eEvent(x,{\bS},f)}

\newcommand{\bM}{\boldsymbol{M}}
\newcommand{\bA}{\boldsymbol{A}}
\newcommand{\obM}{\overline{\bM}}
\newcommand{\bP}{\boldsymbol{P}}
\newcommand{\boldf}{\boldsymbol{f}}
\newcommand{\match}[1]{\ensuremath{\mathcal{M}_{#1}}}
\newcommand{\uinf}{\ensuremath{U_f^-}}
\newcommand{\mumax}{m}

\newcommand{\bmuhat}{\hat{\boldsymbol{\mu}}}
\newcommand{\bnu}{\boldsymbol{\nu}}

\title{Approximating the Distance to Monotonicity of Boolean Functions\thanks{This work was done in part while the authors were visiting the Simons Institute for the Theory of Computing.
A preliminary version of this work appeared in the proceedings of ACM-SIAM Symposium on Discrete Algorithms (SODA), 2020~\cite{PRW20}.
}}
\author{Ramesh Krishnan S. Pallavoor\thanks{Department of Computer Science, Boston University. Email: {\sf rameshkp@bu.edu, sofya@bu.edu}. The work of these authors was partially supported by NSF award CCF-1909612.}
\and Sofya Raskhodnikova\footnotemark[\value{footnote}]
\and Erik Waingarten\thanks{Department of Computer Science, Columbia University. Email: {\sf eaw@cs.columbia.edu}. This work is supported in part by NSF Graduate Research Fellowship (Grant No. DGE-16-44869).}}
\date{}

\begin{document}

\maketitle

\begin{abstract}
We design a nonadaptive algorithm that, given oracle access to a
function $f\colon \{0,1\}^n \to \{0,1\}$ which is $\alpha$-far from monotone, makes poly$(n, 1/\alpha)$ queries and returns an estimate that, with high probability, is an $\widetilde{O}(\sqrt{n})$-approximation to the distance of $f$ to monotonicity. The analysis of our algorithm relies on an improvement to the directed isoperimetric inequality of Khot, Minzer, and Safra (SIAM J. Comput., 2018).
Furthermore, we rule out a poly$(n, 1/\alpha)$-query nonadaptive algorithm that approximates the distance to monotonicity significantly better by showing that, for all constant $\kappa > 0,$ every nonadaptive $n^{1/2 - \kappa}$-approximation algorithm for this problem requires $2^{n^\kappa}$ queries.
This answers a question of Seshadhri (Property Testing Review, 2014) for the case of nonadaptive algorithms.
We obtain our lower bound by proving an analogous bound for erasure-resilient (and tolerant) testers.
Our method also yields the same lower bounds for unateness and being a $k$-junta.
\end{abstract}


\section{Introduction}\label{sec:intro}
Property testing~\cite{RS96,GGR98} was introduced to provide a formal model for studying algorithms for massive datasets. For such algorithms to achieve their full potential, they have to be robust to adversarial corruptions in the input. Tolerant property testing~\cite{PRR06} and, equivalently\footnote{A distance approximation algorithm can be easily converted to a tolerant tester and vice versa, with at most a logarithmic increase in the query complexity. See~\cite{PRR06} and Section~\ref{intro:connection} for a discussion of the relationship.},
distance approximation, generalize the standard property testing model to allow for errors in the input.

In this work, we study
the problem of approximating the distance to
several properties of Boolean functions, with the focus on monotonicity.  A function $f\colon\hypercube\to\{0,1\}$ is {\em monotone} if $f(x) \leq f(y)$ whenever $x \preceq y$, i.e., $x_i \leq y_i$ for all $i \in [n]$.
The (relative) distance between two functions over $\hypercube$ is the fraction of the domain points on which they differ.
Given a function $f$ and a set $\cal P$ (of functions with the desired property), the distance from $f$ to $\cal P$, denoted $\dist(f, \cal P)$, is the
distance from $f$ to the closest function in $\cal P$. Given $\alpha\in(0,1/2)$, a function is $\alpha$-far from $\cal P$ if
$\dist(f, \cal P)\geq \alpha$; otherwise, it is $\alpha$-close. The distance of a function $f$ to monotonicity is denoted \epsf. We study randomized algorithms which, given oracle access to a Boolean function, output an approximation to \epsf by making only a small number of queries.
Specifically, given an input function $f\colon \hypercube \to \{0,1\}$ which is promised to be at least $\alpha$-far from monotone, an algorithm that achieves a $c$-approximation for $c > 1$ should output a real number $\hat{\eps} \in (0, 1)$ that satisfies, with probability at least $2/3$,
\[\epsf \leq \hat{\eps} \leq c \cdot \epsf.\]
Our goal is to understand the best approximation ratio $c$ that can be achieved\footnote{An equivalent way of stating this type of results is to express the approximation guarantee in terms of both multiplicative and additive error, but with no lower bound on the distance. Purely multiplicative approximation would require correctly identifying inputs with the property, which generally cannot be achieved in time sublinear in the size of the input.} in time polynomial in the dimension $\n$ and $1/\alpha$.

\begin{definition}[Hypercube, edge, $i$-edge, decreasing edge]\label{def:hypercube}
The points of the domain $\hypercube$ can be viewed as vertices of an $n$-dimensional hypercube. Two vertices $x,y\in\hypercube$ form an edge $\{x,y\}$ of the hypercube  if they differ in exactly one coordinate. An edge $\{x,y\}$ is an $i$-edge (or an {\em edge along dimension} $i$) if $x$ and $y$ differ only in their $\ord{i}$ coordinates, that is, $x_i \neq y_i$, but $x_j = y_j$ for all $j\in [n] \setminus \{i\}$. An $i$-edge is {\em decreasing} with respect to a function $f$ if $x_i<y_i$ but $f(x)>f(y)$.
\end{definition}

Fattal and Ron~\cite{FR10} investigated a more general problem of approximating the distance to monotonicity of functions on the hypergrid $[t]^\n$. They gave several
algorithms which achieve an approximation ratio $O(\n)$ in time polynomial in $\n$ and $1/\alpha$; for better approximations, they designed an algorithm with the approximation ratio $\n/k,$ for every $k,$ but with running time exponential in~$k$.
It follows from early works on monotonicity testing that for the special case of the hypercube domain, an $O(n)$-approximation
can be obtained by simply estimating the fraction $\nu_f$ of edges that are decreasing with respect to $f$. 
For a Boolean function $f\colon \hypercube \to \{0,1\}$, as shown in~\cite{DGLRRS99,Ras99,GGLRS00,FLNRRS02}, $\epsf/\n\leq\nu_f\leq 2\epsf$. Thus, by obtaining a constant-factor approximation to the fraction of decreasing edges, one gets an $O(n)$-approximation to $\epsf$.

Prior to our work, no nontrivial hardness results were known for  approximating the distance to monotonicity, other than the corresponding lower bounds on (standard) property testing.

\subsection{Our Contributions}

 All our results are on {\em nonadaptive} algorithms. An algorithm is {\em nonadaptive} if it makes all of its queries in advance, before receiving any answers; otherwise, it is {\em adaptive}. Nonadaptive algorithms are especially straightforward to implement and achieve maximal parallelism.
Additionally, every nonadaptive algorithm that approximates the distance to monotonicity of Boolean functions can be easily converted to an algorithm for approximating the $L_p$-distance to monotonicity of real-valued functions~\cite{BerRY14}.

\subsubsection{Approximating the Distance to Monotonicity}

We design a nonadaptive $\widetilde{O}(\sqrt{\n})$-approximation algorithm for the distance to monotonicity
that runs in time polynomial in the number of dimensions, $\n,$ and $1/\alpha$.
Our algorithm improves on the $O(\n)$-approximation obtained by Fattal and Ron~\cite{FR10}.

\begin{theorem}[Approximation Algorithm]\label{thm:alg}
There is a nonadaptive (randomized) algorithm that, given a parameter $\alpha\in(0,1/2)$ and oracle access to a Boolean function $f \colon \hypercube \to \{0,1\}$ which is $\alpha$-far from monotone, makes $\poly(n, 1/\alpha)$  queries and returns an estimate that, with probability at least $2/3$, is an $\widetilde{O}(\sqrt{n})$-approximation to $\epsf$.
\end{theorem}

Our algorithm works by estimating (in addition to the fraction $\nu_f$ of decreasing edges) the size of a particular matching from a class of matchings parameterized by subsets $S \subseteq [n]$ and  consisting of decreasing edges along the dimensions in $S$. For every $S$, the size of a matching 
of decreasing edges, divided by $2^{n}$, is a lower bound for $\epsf$. This is because any monotone function $g\colon \{0,1\}^n \to \{0,1\}$ must disagree with $f$ on at least one endpoint of each decreasing edge. The important feature of this class of matchings is that the
membership of a given edge in a specified matching can be verified locally by querying $f$ on the endpoints of the edge and their neighbors. To analyze our algorithm, we use a slightly improved version of the (robust) directed isoperimetric inequality by Khot et al.~\cite{KMS15}. Our improvements to isoperimetric inequalities of~\cite{KMS15} are discussed in Section~\ref{sec:intro-KMS} and stated in Theorems~\ref{thm:kms} and~\ref{thm:kms-tal}.
Intuitively, if the input function $f$ is $\widetilde{O}(\eps/\sqrt{n})$-close to monotone, then both the fraction of violated edges and the normalized sizes of all matchings considered by our algorithm are below some threshold $\widetilde{\Theta}(\eps/\sqrt{n})$.
We use Theorem~\ref{thm:kms} to show (in Lemma~\ref{lem:event-occurs}) that if $f$ is $\eps$-far from monotone, then either $\nu_f$ is above the threshold or our algorithm is likely to sample some set $S \subseteq [n]$ where the corresponding matching
has normalized size $\widetilde{\Omega}(\eps / \sqrt{n})$.
This allows us to get an $\widetilde{O}(\sqrt n)$-approximation to the distance to monotonicity.

\subsubsection{Improvements in the Isoperimetric Inequalities from~\texorpdfstring{\cite{KMS15}}{Khot et al. (2018)}}\label{sec:intro-KMS}
 In Section~\ref{sec:improving-thm1.9}, we give a sketch of the proof of slightly improved versions of the isoperimetric inequalities of Khot et al.~\cite[Theorems~1.6 and 1.9]{KMS15}. The improved version of~\cite[Theorem~1.6]{KMS15} is stated next.

For all $x\in\hypercube$, define the {\em negative influence} $I_f^-(x)$ to be equal to $0$ if $f(x)=0$ and equal to the number of decreasing edges incident on $x$ if $f(x) = 1$. Note that each decreasing edge ``counts'' towards the endpoint with the function value 1.

\begin{theorem}[Improvement of Theorem~1.6 from~\cite{KMS15}]
\label{thm:kms-tal}
For every function $f\colon \hypercube \to \{0,1\}$,
\begin{align}
\Ex_{\bx \sim \hypercube} \left[\sqrt{I_f^-(\bx)} \right] = \Omega(\epsf). \label{eq:kms-tal}
\end{align}
\end{theorem}

The improved version of~\cite[Theorem~1.9]{KMS15} is Theorem~\ref{thm:kms} that appears later. It is a {\em robust} version of Theorem~\ref{thm:kms-tal}. Specifically, it holds for every 2-coloring of the edges of $f$. The color of each decreasing edge indicates whether it should be ``counted'' towards the endpoint with value 1 or the endpoint with value 0, and the negative influence is generalized accordingly to the {\em robust} version.

Theorems~1.6 and~1.9 in~\cite{KMS15} state that the left-hand side of \eqref{eq:kms-tal} (and similarly its counterpart \eqref{eq:kms} in Theorem~\ref{thm:kms}) is  $\Omega(\frac{\epsf}{\log n + \log(1/\epsf)})$.
We slightly modify the proof of~\cite{KMS15} to get a stronger lower bound of $\Omega(\epsf)$.
Using the original, weaker inequality for our algorithm would result in an approximation to the distance to monotonicity within a multiplicative factor of $\sqrt{n} \cdot\poly(\log n, \log(1/\alpha))$. This would mean that our algorithm is an $\widetilde{O}(\sqrt{n})$-approximation only if $\epsf \geq{1}/{2^{\poly(\log(n))}}$.

\subsubsection{Lower Bounds for Monotonicity, Unateness, and Being a \texorpdfstring{$k$}{k}-Junta}

We show that a slightly better approximation for the distance to monotonicity, specifically, with a ratio of $n^{1/2-\kappa}$ for an arbitrarily small constant $\kappa>0,$ requires exponentially many queries in $n^{\kappa}$ for every nonadaptive algorithm.

\begin{theorem}[Approximation Lower Bound]\label{thm:mono-approx-lb}
For all constant $\kappa\in(0,1/2),$ there exist $\alpha=O(1/n^{1-\kappa})$ and $\eps=\Omega(1/\sqrt n)$ (that is, with $\frac \eps \alpha=\Omega(n^{1/2-\kappa})$),  for which every nonadaptive algorithm requires more than $2^{n^\kappa}$ queries to $f\colon\hypercube\to\{0,1\}$ to distinguish functions $f$ that are $\alpha$-close to monotone from those that are $\eps$-far
from monotone with probability at least 2/3.
\end{theorem}

This result, in combination with Theorem~\ref{thm:alg}, answers an open question on the problem of approximating the distance to monotonicity by Seshadhri~\cite{Ses14} for the case of nonadaptive algorithms. It is the first lower bound for this problem, and it rules out nonadaptive  algorithms that achieve approximations substantially better than $\sqrt{n}$ with $\poly(n, 1/\alpha)$ queries, demonstrating that Theorem~\ref{thm:alg}  is essentially tight.
This bound is exponentially larger than the corresponding lower bound in the standard property testing model and, in fact, than the running time of known algorithms for testing monotonicity. We elaborate on this point in the discussion below on separation.

To obtain Theorem~\ref{thm:mono-approx-lb}, we investigate a variant of the property testing model, called {\em erasure-resilient} testing. This variant, proposed by Dixit et al.~\cite{DRTV18}, is intended to study property testing in the presence of adversarial erasures. An erased function value is denoted by $\bot$. An $\alpha$-erasure-resilient $\eps$-tester for a desired property gets oracle access to a function $f\colon\hypercube\to {\{0,1,\bot\}}$ that has at most an $\alpha$ fraction of values erased. The tester has to accept (with probability at least 2/3) if the erasures can be filled in
to ensure that the resulting function has the property and to reject (with probability at least 2/3) if every completion of erasures results in a function that is $\eps$-far from having the property. As observed in~\cite{DRTV18}, the query complexity of problems in this model lies between their complexity in the standard property testing model and the tolerant testing model. Specifically, a (standard) $\eps$-tester that, given a parameter $\eps,$ accepts functions with the property and rejects functions that are $\eps$-far from the property (with probability at least 2/3), is a special case of an $\alpha$-erasure-resilient $\eps$-tester with $\alpha$ set to 0. Importantly for us,
a tolerant tester that, given $\alpha,\eps\in(0,1/2)$ with $\alpha<\eps$, accepts functions that are $\alpha$-close and rejects functions that are $\eps$-far (with probability at least 2/3) can be used to get an $\alpha$-erasure-resilient $\eps$-tester. The erasure-resilient tester can be obtained by simply filling in erasures with arbitrary values and running the tolerant tester.
We prove a lower bound for erasure-resilient monotonicity testing.

Our method yields lower bounds for two other properties of Boolean functions: unateness, a natural generalization of monotonicity, and being a $k$-junta. A Boolean function $f \colon \hypercube \to \{0,1\}$ is {\em unate} if, for every variable $i\in [\n]$, the function is nonincreasing or nondecreasing in that variable.
A function $f\colon \hypercube\to\{0,1\}$ is a {\em $k$-junta} if it depends on at most $k$ (out of $n$) variables.

We prove the following result on erasure-resilient testing which implies Theorem~\ref{thm:mono-approx-lb}.

\begin{theorem}[Erasure-Resilient Lower Bound]\label{thm:mono-ER-lb}
For all constant $\kappa\in(0,1/2),$ there exist $\alpha=O(1/n^{1-\kappa})$ and $\eps=\Omega(1/\sqrt n)$ (that is, with $\frac \eps \alpha=\Omega(n^{1/2-\kappa})$), for which every nonadaptive $\alpha$-erasure-resilient $\eps$-tester requires more than $2^{n^\kappa}$ queries to test monotonicity of  functions $f\colon\hypercube\to\{0,1\}$. The same bound holds for testing unateness and the $n/2$-junta property.
\end{theorem}

Theorem~\ref{thm:mono-ER-lb} directly implies lower bounds analogous to the one stated in Theorem~\ref{thm:mono-approx-lb} for unateness and being an $n/2$-junta. Lower bounds for approximating the distance to unateness and to being a $k$-junta have been investigated by Levi and Waingarten~\cite{LW19}. They showed that every algorithm approximating the distance to unateness within a constant factor requires $\widetilde{\Omega}(\n)$ queries and strengthened their lower bound to $\widetilde{\Omega}(\n^{3/2})$ queries for nonadaptive algorithms. They also showed that every nonadaptive algorithm that provides a constant approximation to the distance to being a $k$-junta must make $\widetilde{\Omega}(k^{2})$ queries.
Our lower bounds are exponentially larger than those obtained by Levi and Waingarten~\cite{LW19} and hold for larger approximation factors.

\subsubsection{Separation}
Our lower bounds provide natural properties for which erasure-resilient property testing (and hence, distance approximation) is exponentially harder than standard property testing with nonadaptive algorithms. Previously, such strong separation was only known for artificially constructed properties based on PCPs of proximity~\cite{FF06,DRTV18}. For testing monotonicity of Boolean function, the celebrated nonadaptive algorithm of Khot, Minzer and Safra~\cite{KMS15} makes $\widetilde{O}(\sqrt{\n}/\eps^2)$ queries. Unateness can be tested nonadaptively with
 $O(\frac \n \eps \log \frac \n \eps)$ queries~\cite{BCPRS17} whereas the property of being a $k$-junta can be tested nonadaptively with $\widetilde{O}(k^{3/2}/\eps)$ queries~\cite{Bla08}. Our lower bound shows that, for all three properties, nonadaptive testers requires exponentially many queries when the ratio $\eps/\alpha$ is substantially smaller than $\sqrt{n}.$ This stands in contrast to examples of many properties provided in~\cite{DRTV18}, for which erasure-resilient testers have essentially the same query complexity as standard testers.

\subsubsection{Connection Between Tolerant Testing and Distance Approximation}\label{intro:connection}
Our main algorithm presented in Section~\ref{sec:mono-alg} distinguishes functions that are $\widetilde{O}(\eps/\sqrt{n})$-close to monotone from functions that are $\eps$-far from monotone.
In property testing terminology, such an algorithm is an example of a \emph{tolerant tester}. Transforming a tolerant tester to a distance approximation algorithm can be done using standard techniques; see, for example~\cite[Claim~2]{PRR06} and~\cite[Section~3.3]{ACCL07}.
Note that the distance from any Boolean function to monotonicity is at most $1/2$, since the constant-0 and constant-1 functions are monotone.
Therefore, to obtain an algorithm that approximates the distance to monotonicity for an input function $f$ under the assumption $\eps_f\geq \alpha$, we can run the tolerant tester with $\eps$ set to $\frac{1}{2}, \frac{1}{4}, \frac{1}{8}, \ldots, \alpha$ an appropriate number of times. The guarantees of the resulting conversion from tolerant testing to distance approximation are stated in Theorem~\ref{thm:tolerant-to-dist-approx},
which is a straightforward generalization of~\cite[Claim~2]{PRR06}, with a small improvement in query complexity (by a factor of $\log (1/\alpha)$).
This improvement applies to every tolerant tester with query complexity at least linear in $1/\eps$.
The details appear in Section~\ref{sec:tolerant-to-dist-approx}.

\subsection{Comparison to Potential Alternative Approaches to Approximating \texorpdfstring{$\eps_f$}{epsilon\_f}}
Chakrabarty and Seshadhri, in a personal communication, notified us of an alternative approach towards an $O(\sqrt{n})$-approximation
of $\eps_f$ via estimating the size of a maximal matching of decreasing edges. Results in~\cite{KMS15,CS16} imply that the size of a maximal matching is an $O(\sqrt{n})$-approximation to the distance to monotonicity, and there are sublinear-time algorithms for approximating this quantity~\cite{YYI09,ORRR12}. However, these algorithms are adaptive. It is a compelling open problem to understand whether adaptivity can help with approximating the distance to monotonicity.

Given Theorem~\ref{thm:kms-tal}, a natural approach to design an algorithm for approximating the distance to monotonicity is to estimate the left-hand side of \eqref{eq:kms-tal} by  sampling points from $\hypercube$ uniformly at random. When the underlying function $f\colon \{0,1\}^n\to \{0,1\}$ is $\eps$-far from monotone, the estimator would be $\Omega(\eps)$.
The problem is that this estimator could be just as high for functions that are close to monotone. For example, consider the random function $\boldf \colon \hypercube \to \{0,1\}$ which is defined as follows:
\begin{align*}
\boldf(x) =
\begin{cases}
1 - \maj(x) & \text{with probability } \eps/\n; \\
\maj(x) & \text{with probability } 1 - \eps/\n,
\end{cases}
\end{align*}
where $\maj(\cdot)$ denotes the majority function.
Then, $\eps_{\boldf} = \Theta(\eps / n)$ with high probability, yet the left-hand side of \eqref{eq:kms-tal} is $\Omega(\eps)$ with high probability.

Finally, as thoroughly explained in \cite[Section~3.2]{PRR06}, every $q(n,\eps)$-query algorithm for (standard) property testing whose queries are individually (almost) uniform exhibits some degree of tolerance. The algorithm of Khot et al.~\cite{KMS15} for testing monotonicity of Boolean functions makes $\widetilde{O}(\sqrt{n}/\eps^2)$ queries which are roughly uniformly distributed (this is implicit in Lemma~9.3 of \cite{KMS15}), and hence can be used to obtain an $\widetilde{O}(\sqrt{n}/\alpha)$-approximation. However, notice that the approximation factor degrades as a function of $\alpha$.

\subsection{Previous Work}

Testing monotonicity and unateness (first studied in \cite{GGLRS00}), as well as $k$-juntas (first studied in \cite{FKRSS04}), are among the most widely investigated problems in property testing (\cite{EKKRV00,DGLRRS99,Ras99,LR01,FLNRRS02,AC06, Fis04, HK08, BRW05, PRR06, ACCL07, BGJRW12, BCGM12,BBM12, CS13, CS14, CS16, BlaRY14,CDJS17, CST14, CDST15, KMS15, BB16, CWX17a, PRV17, BCS18, CS19, BCS20} study monotonicity testing, \cite{KS16, BCPRS17, CWX17a, CWX17b, CW19} study unateness testing, and \cite{CG04, Bla08, Bla09, BGSMW13, STW15,ChenSTWX18,Sag18} study $k$-junta testing).
Nearly all the previous work on these properties is in the standard testing model. The best bounds on the query complexity of these problems are an $\widetilde{O}(\sqrt{n})$-query nonadaptive algorithm of \cite{KMS15} and lower bounds of $\widetilde{\Omega}(\sqrt{\n})$ (nonadaptive) and $\widetilde{\Omega}(n^{1/3})$ (adaptive) \cite{CWX17a} for monotonicity, and tight upper and lower bounds of $\widetilde{\Theta}(n^{2/3})$ for unateness testing \cite{CW19, CWX17a}, as well as $\Theta(k \log k)$ for $k$-junta testing \cite{Bla08, Sag18}.

Beyond the (standard) property testing, the questions of erasure-resilient and tolerant testing have also received some attention (\cite{DRTV18,RRV19,LeviPRV21} study the erasure-resilient model, and \cite{GR05,PRR06,FF06,FN07,ACCL07,KS09,MR09,FR10,CGR13,BerRY14, BMR16, Tel16, BCELR16, LW19, DMN19} study the tolerant testing model). Specifically for monotonicity, in~\cite{DRTV18}, an erasure-resilient tester for functions on hypergrids is designed. For the special case of the hypercube domain, it runs in time $O(\n/\eps)$ and works when $\eps/\alpha=\Omega(\n)$. Using the connection between distance approximation and erasure-resilient testing, our approximation algorithm implies an erasure-resilient tester that has a less stringent restriction on $\eps/\alpha$, specifically, $\Omega(\sqrt{\n})$. For approximating the distance to $k$-juntas \cite{BCELR16, DMN19}, the best algorithm with additive error 
of $\eps$ makes $2^k \cdot \poly(k,1/\eps)$ queries~\cite{DMN19}, and the best lower bound was $\Omega(k^2)$ for nonadaptive algorithms~\cite{LW19}.

\section{An Approximation Algorithm for the Distance to Monotonicity}\label{sec:mono-alg}

This section is devoted to proving Theorem~\ref{thm:alg}.
We provide a nonadaptive algorithm that gets a parameter $\alpha > 0$ and oracle access to a function $f \colon \hypercube \to \{0,1\}$ promised to be $\alpha$-far from monotone, makes $\poly(n, 1/\alpha)$ queries, and returns an estimate $\hat{\eps} > 0$ that satisfies, with probability at least $2/3,$
\[ \epsf \leq \hat{\eps} \leq \widetilde{O}(\sqrt{n}) \cdot \epsf .\]

Our main algorithm, \ApproxMono, whose performance is summarized in Lemma~\ref{lem:approx-mono}, distinguishes functions that are close to monotone from those that are far.
Theorem~\ref{thm:alg} follows directly from Lemma~\ref{lem:approx-mono} and Theorem~\ref{thm:tolerant-to-dist-approx}.

\begin{lemma}\label{lem:approx-mono}
There exist a universal constant $c_0 \in (0, 1)$
and a nonadaptive algorithm, \eApproxMono, that gets a parameter $\eps \in (0,1/2)$ and oracle access to a function $f\colon \{0,1\}^n \to \{0,1\}$, makes $\poly(n, 1/\eps)$  queries and outputs \emph{\close}\ or \emph{\far}\ as follows:
\begin{enumerate}
\item If $\epsf \leq \appfactorc$, it outputs \emph{\close}\ with probability at least $2/3$.
\item If $\epsf \geq \eps$, it outputs \emph{\far}\ with probability at least $2/3$.
\end{enumerate}
\end{lemma}

The algorithm \ApproxMono\ is described in Algorithm~\ref{alg:approx-mono}. In the algorithm and its analysis, we use $\af$ to denote $\appfactor$.
The algorithm uses subroutines \EdgeViolations\ and \EventEstimation, which are described in Algorithms~\ref{alg:edge-violations} and~\ref{alg:event-estimation}, respectively. The subroutine \EV\ gets a parameter $\delta > 0$ and oracle access to a  function $f\colon \hypercube \to \{0,1\}$, and returns an estimate of the fraction of decreasing edges of $f$ with an additive error less than~$\delta$.
The second subroutine, \EE, gets a parameter $\delta > 0$, a subset $S \subseteq [n]$ and oracle access to a  function $f\colon \hypercube \to \{0,1\}$. The goal of \EE\ is to estimate the probability\footnote{We use the convention that random variables are boldface whereas fixed quantities use standard typeface. For example, $\bx$ and $\bS$ are random variables whereas $x$ and $S$ are the corresponding fixed quantities. The notation $\sim$ stands for ``sampled from''.}, over $\bx \sim \hypercube$, of an event (which we denote $\Event$ and describe in Definition~\ref{def:event}) defined with respect to {$\bx$}, $S$ and $f$ up with an additive error less than $\delta$. The high level intuition is that, as long as the estimates returned by \EE\ and \EV\ are accurate, we can certify a lower bound on the distance to monotonicity of $f$. We prove that if $f$ is $\eps$-far from monotone, either the number of decreasing  edges of $f$ is large (and Step~\ref{step:far-edge-viol} of Algorithm~\ref{alg:approx-mono} declares \far) or the $\EventEstimation$ subroutine can verify a lower bound on the distance to monotonicity.

\begin{algorithm}
\caption{$\ApproxMono(\eps, f)$} \label{alg:approx-mono}
\SetKwInOut{Input}{input} \SetKwInOut{Output}{output}
\SetKwFor{RepeatTimes}{repeat}{times}{end}

\Input{A parameter $\eps \in (0, 1/2)$; oracle access to $f\colon \hypercube \to \{0,1\}$.}
\Output{Either \close\ or \far.}

\DontPrintSemicolon
\BlankLine

\nl Set $\hat{\bnu} \leftarrow \EdgeViolations(\af/4,f)$. \tcp{$\hat{\bnu}$ is an estimate of the fraction of decreasing edges with an additive error less than $\af/4$.}
\nl \label{step:far-edge-viol}\lIf{$\hat{\bnu} \geq 3\af/4$}{\Return{\emph{\far}}.}
\nl \label{step:set-m}Set $t \leftarrow 2/\mumax$, where $\mumax=c_1\cdot\af/2$ with the constant $c_1$ dictated by Lemma~\ref{lem:event-occurs}. \;
\nl \ForEach{$d \in \{1, 2, 4,\dots, 2^{\floor{\log_2 n}}\}$}{
\nl \RepeatTimes{$t$}{
\nl \label{step:sample-S}Sample $\bS \subseteq [n]$ by including each $i \in [n]$ independently with probability $1/d$. \;
\nl Set $\bmuhat \leftarrow \EventEstimation(\bS, \mumax/4, f)$.
\tcp{ $\bmuhat$ is an estimate of $\mu_f(\bS)=\Prx_{\bx \sim \{0,1\}^n}\left[\Event(\bx,\bS, f)\right]$ with an additive error less than $\mumax/4$.}
\nl \lIf{$\bmuhat \geq 3\mumax/4$}{\Return{\emph{\far}}. \label{step:far-event-est}}
}
}
\nl \Return{\emph{\close}}.

\end{algorithm}

Recall the definition of an edge and a decreasing edge of the hypercube (see Definition~\ref{def:hypercube}).

\begin{definition}\label{def:decreasing-edges}
For a function $f:\{0,1\}^n\to\{0,1\}$, let $\nu_f$ denote the probability that a uniformly random edge is decreasing.
We also refer to $\nu_f$ as the fraction of decreasing edges.
\end{definition}
For a dimension $i \in [\n],$ a point $x \in \hypercube,$ and a bit $b \in \{0,1\}$, we use $x^{(i \to b)}$ to denote the point in $\hypercube$ whose $\ord{i}$ coordinate is $b$ and the remaining coordinates are the same as in $x$. We use $x^{(i)}$ to denote the point $x^{(i \to (1-x_i))}$, where $x_i$ is the $\ord{i}$ coordinate in $x$.

We summarize the properties of the subroutine \EV, which estimates $\nu_f$, in Fact~\ref{fact:edge-estimation}.

\begin{algorithm}
\caption{$\EdgeViolations(\delta, f)$} \label{alg:edge-violations}
\SetKwInOut{Input}{input} \SetKwInOut{Output}{output}
\SetKwFor{RepeatTimes}{repeat}{times}{end}

\Input{A parameter $\delta > 0$; oracle access to $f\colon \hypercube \to \{0,1\}$.}
\Output{A real number $\hat{\bnu} \in [0,1]$ that approximates $\nu_f$.}

\DontPrintSemicolon
\BlankLine

\nl Initialize $\bc \gets 0$ and set $t \gets$ $\ceil{\frac{2}{\delta^2}}$.\; 
\nl \RepeatTimes{$t$}{
	\nl Sample $\bx \sim \{0,1\}^n$ and $\bi \sim [n]$, and query $f(\bx^{(\bi\to 0)})$ and $f(\bx^{(\bi\to 1)})$.\;
	\nl \lIf{$f(\bx^{(\bi\to0)}) > f(\bx^{(\bi\to 1)})$}{$\bc\gets \bc+1$.} 
}
\nl \Return{$\hat{\bnu} = \bc/t$}.

\end{algorithm}

\begin{algorithm}
\caption{$\EventEstimation(S, \delta, f)$} \label{alg:event-estimation}
\SetKwInOut{Input}{input} \SetKwInOut{Output}{output}
\SetKwFor{RepeatTimes}{repeat}{times}{end}

\Input{A set $S \subseteq [n]$; a parameter $\delta > 0$; oracle access to $f\colon \hypercube \to \{0,1\}$.}
\Output{A real number $\bmuhat \in [0,1]$ that approximates $\mu_f(S)$.}

\DontPrintSemicolon
\BlankLine

\nl Initialize $\bc \gets 0$ and set $t \gets \left\lceil\frac{2\ln(\n/\eps)}{\delta^2}\right\rceil$.\;
\nl \RepeatTimes{$t$}{
	\nl Sample $\bx \sim \{0,1\}^n$ and query $f(\bx)$.\;
	\nl \ForEach{$i \in S$}{
		\nl Set $\by_i \gets \bx^{(i)}$ and query $f(\by_i)$.\;
		\nl \lForEach{$j \in S \setminus \{i\}$}{query $f(\by_i^{(j)})$.}
	}
	\nl {\bf if} for some $i \in S$, the edge $\{\bx,\by_i\}$ is decreasing and, for all $j \in S \setminus \{i\}$, the edge $\{\by_i, \by_i^{(j)}\}$ is nondecreasing, {\bf then} $\bc\gets \bc+1$. \tcp{$\Event(\bx, S, f)$ occurred.}
}
\nl \Return{$\bmuhat = {\bc}/{t}$}.
\end{algorithm}

\begin{fact}\label{fact:edge-estimation}
The algorithm $\eEdgeViolations$ is nonadaptive. It gets a parameter $\delta > 0$ and oracle access to a  function $f\colon \hypercube \to \{0,1\}$, makes $O(1/\delta^2)$ queries, and outputs $\hat{\bnu} \in [0,1]$ which, with probability at least $24/25$, satisfies
\begin{align}
|\nu_f - \hat{\bnu} | < \delta. \label{eq:violating-edges-estimator}
\end{align}
\end{fact}

\begin{proof}
By Hoeffding bound, $\Pr[|\nu_f - \hat{\bnu} | \geq \delta]\leq 2 \eee^{-2t\delta^2}\leq 2\eee^{-4}\leq 1/25.$
\end{proof}

\begin{observation}\label{obs:edge-lb}
For all functions $f:\{0,1\}^n\to\{0,1\}$, we have $\nu_f\leq 2\epsf.$
\end{observation}
\begin{proof}
By definition, $\nu_f\cdot 2^{n-1}n$ edges in $f$ are decreasing. To convert $f$ to a monotone function, at least one endpoint of every decreasing edge has to be changed.
Since each point in $\hypercube$ is incident on $\n$ edges of the hypercube, at least $\nu_f\cdot 2^{n-1}$ points have to be changed. Thus, $\epsf\geq\nu_f\cdot 2^{n-1}/2^n=\nu_f/2.$
\end{proof}

\begin{definition}\label{def:event}
For a  function $f\colon \{0,1\}^n \to \{0,1\}$, a subset $S \subseteq [n],$ and a point $x \in \{0,1\}^n$, let $\Event(x, S, f)$ be the following event (see Figure~\ref{fig:capture-event}):
\begin{enumerate}
\item There exists an index $i \in S$ such that $\{ x, y \}$ is a decreasing edge in $f$, where $y = x^{(i)}$.
\item For all $j \in S \setminus \{i\}$, the edge $\{ y, y^{(j)} \}$ is a nondecreasing edge in $f$.
\end{enumerate}
We denote  $\Prx_{\bx \sim \{0,1\}^n}\left[\Event(\bx, S, f) \right]$ by $\mu_f(S)$.
\end{definition}

\begin{figure}[h]
\centering
    \includegraphics[scale=.6]{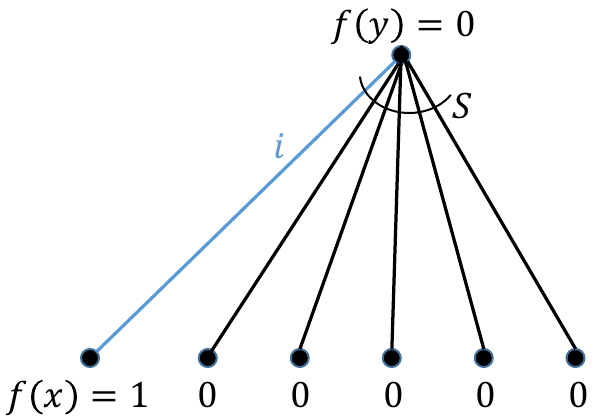}
    \hspace*{1cm}
    \includegraphics[scale=.6]{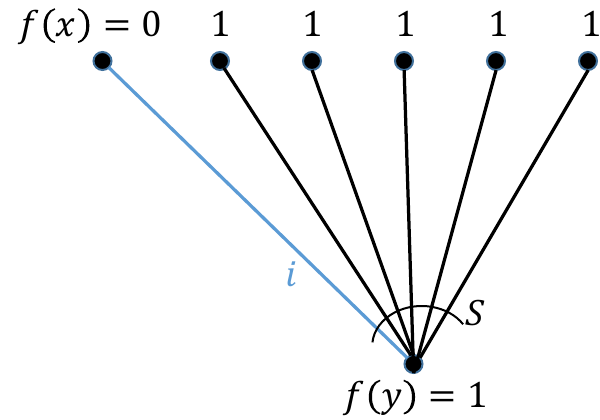}
    \caption{An illustration to Definition~\ref{def:event}. Two cases are depicted, corresponding to the two possible values of $f(x)$.}
    \label{fig:capture-event}
\end{figure}

We summarize the properties of subroutine \EE, which estimates $\mu_f(S)$, in Fact~\ref{fact:event-estimation}. Like Fact~\ref{fact:edge-estimation}, it can be easily proved by an application of Hoeffding bound.

\begin{fact}\label{fact:event-estimation}
The algorithm $\eEventEstimation$ is nonadaptive. It gets a set $S \subseteq [n]$, a parameter $\delta > 0$ and oracle access to a  function $f\colon \hypercube \to \{0,1\}$, makes $O(|S|^2 \log (n/\eps)/ \delta^2)$ queries, and outputs $\bmuhat \in [0,1]$ which, with probability at least $1 - (\eps/n)^3$, satisfies
\begin{align}
|\mu_f(S) - \bmuhat| < \delta. \label{eq:capture-event-estimator}
\end{align}
\end{fact}

Lemma~\ref{lem:approx-mono} follows from Lemmas~\ref{lem:dist-lb} and~\ref{lem:event-occurs}.
\begin{lemma}\label{lem:dist-lb}
For all functions $f\colon \{0,1\}^n \to \{0,1\}$ and sets $S \subseteq [n]$,
\[ \mu_f(S) \leq 2\cdot\epsf. \]
\end{lemma}

\begin{proof}
Let $X = \{ x \in \{0,1\}^n : \text{the event } \Event(x, S, f) \text{ occurs}\}$. For each $x \in X$, let $y_x = x^{(i)}$ for a dimension $i \in S$ be a point for which $\{x, y_x\}$ is decreasing and, for all $j \in S \setminus \{i\}$, the edge $\{y_x, y^{(j)}_x\}$ is nondecreasing. Consider the set of decreasing edges of $f$ given by $E_X = \{ \{x, y_x\} : x \in X\}$. If $x_1, x_2$ from $X$ are distinct, then $y_{x_1} \neq y_{x_2}$\,, because otherwise $y_{x_1}$ would violate Item 2 in Definition~\ref{def:event}. Thus, $E_X$ is a matching. Each edge is added to $E_X$ at most twice (once for each endpoint), so $|E_X| \geq |X|/2$. Since we have a matching of at least $|X|/2$ decreasing edges,  we must change $f$ on at least $|X|/2=\mu_f\cdot 2^n/2$ points to make it monotone.
\end{proof}

\begin{lemma}[Key Lemma]\label{lem:event-occurs}
There exists a constant $c_1 \in (0, 1)$ such that the following holds. Let $f\colon \{0,1\}^n \to \{0,1\}$ be $\eps$-far from monotone, with $\nu_f<\af$.  Let $\mumax_f$ be
the maximum over $d \in \{1,2,4,\dots, 2^{\floor{\log_2 n}}\}$ of
\[\Ex_{\substack{\bS \subseteq [n] \\ i \in \bS \text{ w.p. } 1/d}}\left[\mu_f(\bS)\right].\]
Then $\mumax_f > c_1 \cdot \af$.
\end{lemma}

The proof of Lemma~\ref{lem:event-occurs} appears in Section~\ref{sec:event-occurs}. We use Observation~\ref{obs:edge-lb} and Lemmas~\ref{lem:dist-lb} and~\ref{lem:event-occurs} to complete the proof of Lemma~\ref{lem:approx-mono}.

\begin{proof}[Proof of Lemma~\ref{lem:approx-mono}]
We set the constant $c_0 = c_1 / 8$ in Lemma~\ref{lem:approx-mono}, where the constant $c_1 \in (0,1)$ is 
from Lemma~\ref{lem:event-occurs}.

Let $\CorrectEstimates$ be the event that all invocations of the subroutines in Algorithm~\ref{alg:approx-mono} produce outputs within the error bounds specified in (\ref{eq:violating-edges-estimator}) and (\ref{eq:capture-event-estimator}).
By Facts~\ref{fact:edge-estimation} and~\ref{fact:event-estimation} and by a union bound over the invocation of $\EdgeViolations$ and at most $t (\log_2 n+1) = O(n/\eps)$
invocations of $\EventEstimation$, this event occurs with probability at least $5/6$.

First, we prove part 1 of Lemma~\ref{lem:approx-mono}. Suppose $\epsf \leq c_0\cdot\af$. Then, by Observation~\ref{obs:edge-lb}, $\nu_f\leq 2c_0\cdot\af<\af/2$, since $c_0 \leq 1/8$. By Lemma~\ref{lem:dist-lb} and our choice of $c_0$, for all sets $S\subseteq[n]$, we have
 $\mu_f(S)\leq 2 c_0\cdot\af = c_1 \cdot \af / 4 = \mumax/2,$ where $\mumax$ is the parameter from Step~\ref{step:set-m} of Algorithm~\ref{alg:approx-mono}.
When event $\CorrectEstimates$ occurs, the estimate  $\hat{\bnu}$  in Step~\ref{step:far-edge-viol} is less than  $3\af/4$ and all the estimates $\bmuhat$
in Step~\ref{step:far-event-est} are less than $3\mumax/4$.
In this case, Algorithm~\ref{alg:approx-mono} outputs $\close$. This happens with probability at least 3/4, so the proof of part 1 is complete.

Next, we prove part 2 of Lemma~\ref{lem:approx-mono}.
Suppose that $f \colon \{0,1\}^n \to\{0,1\}$ is $\eps$-far from monotone. If $\nu_f\geq\af$ then, whenever $\CorrectEstimates$ occurs, Step~\ref{step:far-edge-viol} outputs \far. Now, assume $\nu_f<\af$.
By Lemma~\ref{lem:event-occurs}, there exists some $d \in \{1,2,4, \dots, 2^{\floor{\log_2 n}}\}$ for which
$\Ex_{\bS \subseteq [n]}[\mu_f(\bS)]=\mumax_f\geq 2\mumax$,
where $\bS \subseteq [n]$ is sampled by including each $i \in [n]$ independently with probability $1/d$.
By the reverse Markov's inequality, since $\mu_{f}(\bS) \leq 1$ always holds, we get
\[
\Pr_{\bS}\left[\mu_f(\bS)\geq \mumax \right]\geq \frac{\mumax_f-\mumax}{1-\mumax}\geq \mumax.
\]
Step~\ref{step:sample-S} of $\ApproxMono(\eps, f)$ fails to sample some $\bS \subset [n]$ such that $\mu_f(\bS) \geq \mumax$ with probability at most
$(1-m)^t\leq \eee^{-mt}
\leq \eee^{-2}\leq 1/6.
$
Since $\CorrectEstimates$ occurs with probability at least $5/6$, Step~\ref{step:far-event-est} outputs \far\ with probability at least $5/6 - 1/6=2/3$.
\end{proof}

\subsection{Proof of Lemma~\ref{lem:event-occurs}}\label{sec:event-occurs}
We start by outlining some of the ideas from our proof of Lemma~\ref{lem:event-occurs} at a high level. In the proof, we attribute each decreasing edge to its endpoint adjacent to a larger number of decreasing edges. We partition all the hypercube vertices into $\log^2 n$ buckets $B_{d,s}$, where the bucket $B_{d,s}$ contains each vertex with $d$ to $2d$ adjacent decreasing edges and $s$ to $2s$ adjacent decreasing edges attributed to it. Importantly, we show that, for each vertex $x\in B_{d,s}$,   when each coordinate of $[n]$ is included in $\bS$ independently with probability $1/d$, the probability of the event $\EventArgs$ is $\Omega(s/d)$. Then we apply (a variant of) the Cauchy-Schwartz inequality, the directed (robust) isoperimetric inequality, and the upper bound on the fraction $\nu_f$ of decreasing edges (from the premise of Lemma~\ref{lem:event-occurs}) to get a lower bound on
$$\sum_{d\in\{1,2,4,\dots,2^{\floor{\log_2 n}}\}}\Ex_{\substack{\bS \subseteq [n] \\ i \in \bS \text{ w.p.}\; 1/d}}
[\mu_f(\bS)].$$
We conclude that, by averaging, there exists a setting of $d$ for which
one of the terms in the sum is large.

A crucial tool
in our proof of Lemma~\ref{lem:event-occurs}
is the main (robust) directed isoperimetric inequality of Khot et al.~\cite{KMS15}. We use notation consistent with~\cite{KMS15}. For a function $f \colon \{0,1\}^n \to \{0,1\}$, let $\calS_f^-$ denote the set of decreasing edges of $f$. Let the function $\uinf \colon \hypercube \to \{0, 1, \dots, n\}$ map  each point $x \in \{0,1\}^n$ to the number of decreasing edges\footnote{Note that $\uinf$ is the ``undirected'' version of $I_f^-$. Specifically,
\begin{align*}
I_f^-(x) = \begin{cases}
\uinf(x) & \text{ if } f(x)=1, \\
0 & \text{ if } f(x) = 0.
\end{cases}
\end{align*}
} of $f$ incident on $x$.
 For an arbitrary coloring of $\calS_f^-$ into red and blue edges, $\col \colon \calS_f^- \to \{ \red, \blue\}$, let
 $I_{f,\col}^-  \colon \{0,1\}^n \to \{ 0, \dots, n \}$ be the function given by:
\[
I_{f, \col}^-(x) =
	\begin{cases}
		| \{ \{x, y\} \in \calS_f^- : \col(\{x, y\}) = \red \}| & \text{if } f(x) = 1;\\
		| \{ \{x, y\}  \in \calS_f^- : \col(\{x, y\}) = \blue\}| & \text{if } f(x) = 0.
	\end{cases}
\]
That is, each decreasing edge is counted towards the lower endpoint if it is red and towards the higher endpoint if it is blue.

We crucially rely on the main theorem of \cite{KMS15}, which is stated next, with a minor improvement in the bound. We obtain the improvement in Section~\ref{sec:improving-thm1.9}.
\begin{theorem}[Improvement of Theorem 1.9 from \cite{KMS15}]\label{thm:kms}
Let $f\colon \{0,1\}^n \to \{0,1\}$ be $\eps$-far from monotone. Then, for every coloring $\col$ of $\calS_f^-$ into $\red$ and $\blue$,
\begin{align}
\Ex_{\bx \sim \{0,1\}^n}\left[ \sqrt{I_{f, \col}^-(\bx)}\right]  = \Omega\left(\eps\right). \label{eq:kms}
\end{align}
\end{theorem}

To prove Lemma~\ref{lem:event-occurs}, consider a function $f\colon \{0,1\}^n \to \{0,1\}$ which is $\eps$-far from monotone with
$\nu_f<\af$.
Consider the coloring of $\calS_f^-$ that colors each edge $\{x,y\}$ with $x\prec y$ as follows:
\[ \col(\{x, y\}) = \left\{ \begin{array}{cc} \red & \text{ if  } \uinf(x) \geq \uinf(y); \\
						\blue & \text{ if  } \uinf(x) < \uinf(y).
\end{array} \right. \]
In this coloring, each decreasing edge in $f$ is counted in \eqref{eq:kms} towards its endpoint incident on a higher number of decreasing edges. (If there is a tie, it is counted towards the lower endpoint.)

Without loss of generality,
suppose that the red edges contribute at least as much as the blue edges to the Talagrand objective (the case where the blue edges contribute more is symmetric). In other words, we break up the left-hand side of (\ref{eq:kms}) into two terms and assume one is greater than or equal to the other:
\[
\sum_{x:f(x)=1}\left[ \sqrt{I_{f, \col}^-(x)}\right]\geq \sum_{x:f(x)=0}\left[ \sqrt{I_{f, \col}^-(x)}\right].
\]
Then, by (\ref{eq:kms}),
\[
\frac 1{2^n}\sum_{x:f(x)=1}\left[ \sqrt{I_{f, \col}^-(x)}\right]
\geq\frac 1 2\Ex_{\bx \sim \{0,1\}^n}\left[ \sqrt{I_{f, \col}^-(\bx)}\right]  = \Omega\left(\eps\right).
\]
We partition the points $x\in \{0,1\}^n$ with $f(x)=1$ and $I_{f, \col}^-(x)>0$ into buckets $B_{d,s}$ indexed by $d,s\in \{1,2,4,\dots, 2^{\floor{\log_2 n}}\}$, where $d\geq s$. The bucket $B_{d,s}$ is defined by:
\[
B_{d,s}=\{ x \in \{0,1\}^n : d \leq \uinf(x) < 2d \text{ and } s \leq I_{f,\col}^-(x) < 2s \text{ and } f(x) = 1\}.
\]
That is, vertices in $B_{d,s}$ are incident on between $d$ and $2d$ decreasing edges and  between $s$ and $2s$ red edges.
By the definition of the buckets, $I_{f,\col}^-(x)\leq 2s$ for every $x\in B_{d, s}$, implying
\begin{equation}\label{eq:sum-of-sq-roots-lb}
\sum_{d,s:d\geq s}|B_{d,s}|\cdot\sqrt{s}
= \sum_{d,s: d\geq s}\ \ \sum_{x \in B_{d, s}} \sqrt{s}
\geq\frac 1{\sqrt{2}} \sum_{x:f(x)=1} \sqrt{I_{f,\col}^-(x)}
=\Omega(\eps\cdot 2^n).
\end{equation}
Each $x \in B_{d, s}$ is an endpoint of at least $d$ decreasing edges of $f$. Moreover, the sets of decreasing edges incident on different points $x$ with $f(x)=1$  are disjoint. Consequently, by the bound on the fraction of decreasing edges in the statement of Lemma~\ref{lem:event-occurs},
\begin{equation}\label{eq:sum-of-d-ub}
\sum_{d,s:d\geq s}|B_{d,s}|\cdot d<
\af\cdot 2^{n-1}n
=\eps\sqrt{n/\log n}\cdot 2^{n-1}.
\end{equation}

Next, we show that for each bucket $B_{d,s}$ and each $x \in B_{d,s}$, the probability that the event $\EventArgs$ occurs is sufficiently large when $\bS$ is chosen appropriately.

\begin{claim} \label{clm:good-set-event}
For all $d,s\in \{1,2,4,\dots, 2^{\floor{\log_2 n}}\}$, where $d\geq s,$ and all $x \in B_{d,s}$,
\[
\Prx_{\substack{\bS \subseteq [n] \\ i \in \bS \text{ w.p.}\; 1/d}}\left[ \eEventArgs \right] \geq\frac{1}{\eee^4} \cdot \frac{s}{d}.
\]
\end{claim}

\begin{proof}
Fix $d\geq s$ and an arbitrary $x \in B_{d,s}\,$.

First, consider the case when $d = 1$. Then $s = 1$ and, consequently,
\[\uinf(x) = I_{f, \col}^- (x) = 1,\]
that is, the only decreasing edge incident on $x$ is colored $\red$. Call this edge $\{x,y\}.$
Since $\col(x,y)=\red,$ by definition of the coloring, $\uinf(y)\leq \uinf(x)=1.$
Therefore, $x$ and $y$ are not endpoints of any decreasing edges other than the edge $\{x,y\}$.
Note that $\bS = [\n],$ since each $i \in [\n]$ is in $\bS$ with probability $1/d = 1$. By Definition~\ref{def:event}, $\EventArgs$ occurs since $\{x,y\}$ is a decreasing edge along a dimension in $\bS$, and all other edges incident on $y$ are nondecreasing. Hence,
\[
\Pr_{\bS = [\n]}[\EventArgs] = 1,
\]
concluding the proof for the case $d = 1$.

Now, consider the case when $d \geq 2$.  For $x \in \hypercube$, let $\viol{x} = \{i \in [n] : \{x, x^{(i)}\} \in \calS_f^- \}$ denote the set of dimensions along which the 
edges incident on $x$ are decreasing in $f$, and let $\goodviol{x} = \{i \in \viol{x} : \uinf(x) \geq \uinf(x^{(i)})\}$ 
be the set of dimensions along which the other endpoint is adjacent to no more decreasing edges than $x$.
By the definition of the buckets, $|\viol{x}| \leq 2d-1$ and $s\leq |\goodviol{x}| \leq 2s-1$.
If we sample $\bS \subseteq [n]$ by including each index $i \in [n]$ independently with probability $1/d$, then
the probability of $\EventArgs$ is at least the probability that there exists a unique $i \in \bS$ such that $y = x^{(i)}$ satisfies $\{x, y\} \in \calS_f^-$ with $\uinf(x) \geq \uinf(y)$, and all other decreasing edges of $f$ incident on $y$
are along dimensions in $[\n] \setminus \bS$.
This probability, in turn, is at least the probability of the union over $i\in\goodviol{x}$ of the following events: for each  $i\in\goodviol{x}$, the corresponding event is that $i\in \bS$, but all other dimensions in $\goodviol{x} \cup \viol{x^{(i)}}$ are not in $\bS$.
Since these events are disjoint,
\begin{align*}
\Pr_{\bS \subseteq [n]}\left[ \EventArgs \right]
&\geq \sum_{i \in \goodviol{x}} \Big( \Pr[i \in \bS] \cdot
	\prod_{
		j \in (\goodviol{x}\cup \viol{x^{(i)}}) \setminus \{i\}}
		\Pr[j \notin \bS] \Big),\\
& \geq s \cdot \frac{1}{d} \cdot \left(1 - \frac{1}{d} \right)^{(2s-2)+(2d-2)}
 \geq s \cdot \frac{1}{d} \cdot \left(1 - \frac{1}{d} \right)^{4(d-1)}
 \geq \frac{s}{\eee^4d},
\end{align*}
where we used $s\leq |\goodviol{x}| \leq 2s-1$ and $$|\viol{x^{(i)}}| \leq \uinf(x) = |\viol{x}| \leq 2d-1$$ to get the second inequality, $d\geq s$ to get the third inequality, and $(1-1/d)^{d-1} \geq 1/\eee$ for all $d \geq 2$ to get the final inequality.
This concludes the proof of Claim~\ref{clm:good-set-event}.
\end{proof}
By Claim~\ref{clm:good-set-event}, for all $d\in \{1,2,4,\dots, 2^{\floor{\log_2 n}}\},$
\begin{align}
\Ex_{\substack{\bS \subseteq [n] \\ i \in \bS \text{ w.p. } 1/d}}\left[\mu_f(\bS)\right]
&=\Ex_{\substack{\bS \subseteq [n] \\ i \in \bS \text{ w.p.}\; 1/d}}\Big[ \Prx_{\bx \sim \{0,1\}^n}\left[ \Event(\bx, \bS, f)\right]\Big]\nonumber\\
&\geq \frac{1}{2^n} \sum_{x \in \bigcup_s  B_{d,s}} \ \ \Prx_{\substack{\bS \subseteq [n] \\ i \in \bS \text{ w.p.}\; 1/d}}\left[ \EventArgs\right] \nonumber\\
&= \frac{1}{2^n} \sum_{s\in\{1,2,4,\dots,d\}}|B_{d,s}| \cdot \frac{s}{\eee^4d}\;. \label{eq:one-d}
\end{align}
To conclude the proof of Lemma~\ref{lem:event-occurs},
we apply Titu's Lemma, which
states that for all positive real numbers $a_1,\dots,a_k$ and $b_1,\dots,b_k \,,$
\[
\frac{(\sum_{i=1}^k a_i)^2}{\sum_{i=1}^kb_i}\leq \sum_{i=1}^{k}\frac{a_i^2}{b_i}\,.
\]
Titu's lemma follows directly from the Cauchy-Schwartz inequality:
\[
\left(\sum_{i=1}^k a_i\right)^2 = \left(\sum_{i=1}^k \frac{a_i}{\sqrt{b_i}} \cdot \sqrt{b_i}\right)^2 \leq \left(\sum_{i=1}^k \frac{a_i^2}{b_i}\right) \left(\sum_{i=1}^k b_i\right)\,.
\]
By (\ref{eq:one-d}) and Titu's Lemma,
\begin{align*}
\sum_{d\in\{1,2,4,\dots,2^{\floor{\log_2 n}}\}}\Ex_{\substack{\bS \subseteq [n] \\ i \in \bS \text{ w.p.}\; 1/d}}
[
\mu_f(\bS)
]
&\geq \frac{1}{2^n} \sum_{d,s: d\geq s}  |B_{d,s}| \cdot \frac{s}{\eee^4d}\\
= \frac 1{\eee^4}\cdot \frac{1}{2^n} \sum_{\substack{d,s: d\geq s \\  |B_{d,s}| \neq 0}}  \frac{(|B_{d,s}|\sqrt{s})^2}{|B_{d,s}|d}
&\geq \frac 1{\eee^4}\cdot \frac{1}{2^n}   \frac{(\sum_{d,s: d\geq s}|B_{d,s}|\sqrt{s})^2}{\sum_{d,s: d\geq s}|B_{d,s}|d}\\
= \frac{1}{2^n} \cdot  \frac{\Omega((\eps\cdot 2^n)^2)}{\eps\sqrt{n/\log n}\cdot 2^{n-1}}
&=\Omega\left(\frac{\eps\sqrt{\log n}}{\sqrt{n}}\right),
\end{align*}
where we used \eqref{eq:sum-of-sq-roots-lb} and \eqref{eq:sum-of-d-ub} to get the final line. By an averaging argument, since the summation on the left-hand side above has $O(\log n)$ terms, at least one of the terms is $\Omega(\eps/(\sqrt{n\log n}))$, completing the proof of Lemma~\ref{lem:event-occurs}.

\section{Improvements to the Isoperimetric Inequalities from\texorpdfstring{~\cite{KMS15}}{~Khot et al. (2018)}}\label{sec:improving-thm1.9}
In this section, we prove Theorems~\ref{thm:kms} and~\ref{thm:kms-tal}.

First, we set up some notation.
For a function $f\colon\hypercube\to\{0,1\}$, a set $S\subseteq[n]$, and a string $z\in\{0,1\}^{\overline{S}}$, let $f(\cdot, z) \colon \{0,1\}^{S} \to \{0,1\}$ denote the function $f$ restricted to the subcube $\{0,1\}^{S}$ and obtained from $f$ by setting the input bits  in $\{0,1\}^{\overline{S}}$ to $z$.

In Proposition~\ref{prop:kms-bound} used in the proof of~\cite[Theorem~1.6]{KMS15} and stated below, we fix a subset $S \subseteq [n]$ and sample a uniformly random $\bz \sim \{0,1\}^{\overline{S}}$. Then we consider $f(\cdot, \bz)$, a random restriction of $f$.
\begin{proposition} \label{prop:kms-bound}
For a function $f\colon\{0,1\}^n \to \{0,1\}$ and a
set $S \subseteq [n]$,
\[ \Ex_{\bx \sim \{0,1\}^n}\left[ \sqrt{I_{f}^-(\bx)}\right]
\geq \Ex_{\bz \sim \{0,1\}^{\overline{S}}}\left[
\Ex_{\bw\sim\{0,1\}^{S}} \left[\sqrt{I_{f(\cdot, \bz)}^-(\bw)}\right] \right]. \]
\end{proposition}
Intuitively, for the restricted function $f(\cdot, \bz)$, we count the decreasing edges \emph{only along dimensions in} $S$.
To strengthen Proposition~\ref{prop:kms-bound}, instead of fixing $S$, we sample $S$ according to the following distribution: For a real number $p \in [0,1]$, let $\calS(p)$ denote the distribution on subsets $S \subseteq [\n]$, where each $i \in [\n]$ is included in $S$ independently with probability $p$.
We improve Proposition~\ref{prop:kms-bound} to the following.
\begin{proposition} \label{prop:stronger-bound}
For a function $f\colon\{0,1\}^n \to \{0,1\}$ and a parameter $p \in [0, 1]$,
\[
\sqrt{p} \cdot \Ex_{\bx \sim \{0,1\}^n} \left[ \sqrt{I_f^-(\bx)}\right]
\geq \Ex_{\substack{\bS \sim \calS(p) \\ \bz \sim \{0,1\}^{\overline{\bS}}}}\left[ \Ex_{\bw\sim\{0,1\}^{\bS}} \left[\sqrt{I_{f(\cdot, \bz)}^-(\bw)}\right] \right].
\]
\end{proposition}

\begin{proof}
Recall that for $x \in \hypercube$, the set $\viol{x}$ denotes the subset of dimensions along which the edges incident on $x$ are decreasing in $f$.
Let
\begin{align*}
\dviol{x} = \begin{cases}
\viol{x} & \text{ if } f(x)=1, \\
0 & \text{ if } f(x) = 0.
\end{cases}
\end{align*}
Note that $|\dviol{x}| = I_f^-(x)$ for all $x \in \hypercube$. Hence, 
\begin{align*}
\Ex_{\substack{\bS \sim \calS(p) \\ \bz \sim \{0,1\}^{\overline{\bS}}}} \left[ \Ex_{\bw \sim \{0,1\}^{\bS}} \left[\sqrt{I_{f(\cdot, \bz)}^-(\bw)} \right] \right]
& = \Ex_{\substack{\bS \sim \calS(p) \\ \bx \sim \hypercube}} \left[ \sqrt{|\dviol{\bx} \cap \bS|}\right] \\
& \leq \Ex_{\bx \sim \hypercube} \left[ \sqrt{\Ex_{\bS \sim \calS(p)}\left[|\dviol{x} \cap \bS| \right]} \right] \\
& = \Ex_{\bx \sim \hypercube} \left[ \sqrt{I_f^-(x) \cdot p} \right] \\
& = \sqrt{p} \cdot \Ex_{\bx \sim \{0,1\}^n}\left[ \sqrt{I_f^-(\bx)}\right],
\end{align*}
where we used Jensen's inequality and the fact that the function $\phi(t) = \sqrt{t}$ is concave to derive the inequality. 
\end{proof}

Similarly, we have the analogous proposition for the proof of the robust version of the Talagrand objective (Theorem~1.9 of~\cite{KMS15}).

\begin{proposition} \label{prop:improv-tal-obj}
For a function $f\colon\{0,1\}^n \to \{0,1\}$, a coloring $\col \colon \calS_f^- \to \{ \red, \blue\}$, and a parameter $p \in [0, 1]$,
\begin{align*}
\sqrt{p} \cdot \Ex_{\bx \sim \{0,1\}^n}  \Big[ \sqrt{I_{f,\col}^-(\bx)}\Big]
\geq \Ex_{\substack{\bS \sim \calS(p) \\ \bz \sim \{0,1\}^{\overline{\bS}}}} \Bigg[ \Ex_{\bw\sim\{0,1\}^{\bS}} \Big[\sqrt{I_{f(\cdot, \bz), \col}^-(\bw)}\Big] \Bigg].
\end{align*}
\end{proposition}

Khot et al.~\cite{KMS15} established a connection between the Talagrand objective and the function obtained after applying the \emph{switch} operator, which was defined in~\cite{DGLRRS99, GGLRS00} (see also~\cite[Definition~4]{Ras99}).
\begin{definition}[Switch operator]
The \emph{switch} operator with a parameter $i\in[n]$ applied to a function $f:\hypercube \to \{0,1\}$ returns a function $\switch_i [f] \colon \hypercube \to \{0, 1\}$ defined as follows:
\begin{align*}
\switch_i[f](x) =
	\begin{cases}
		\min(f(x^{(i \to 0)}), f(x^{(i \to 1)})) & \text{ if } x = x^{(i \to 0)}; \\
		\max(f(x^{(i \to 0)}), f(x^{(i \to 1)})) & \text{ if } x = x^{(i \to 1)}.
	\end{cases}
\end{align*}
The definition of the switch operator $\switch$ can be extended to an ordered set $T = (i_1, \ldots, i_\ell)$, where $i_1, \ldots, i_\ell \in [n]$, by applying the switch operator along the dimensions in $T$ in order:
\[
\switch_T [f] = \switch_{i_\ell}[\switch_{i_{\ell-1}}[\cdots \switch_{i_2}[\switch_{i_1}[f]] \cdots]].
\]
\end{definition}

Recall from Definition~\ref{def:hypercube} that, for $i \in [\n]$, an $i$-edge is a hypercube edge whose endpoints differ in the $\ord{i}$ coordinate.
Intuitively, the switch operator ``repairs'' each decreasing $i$-edge by switching the function values on its endpoints. Note that the switch operator has no effect on all nondecreasing $i$-edges.
Dodis et al.~\cite{DGLRRS99} proved that,
for every permutation $\rho$ of $[\n]$, the function $\switch_\rho[f]$ is monotone. Consequently, $\dist(f, \switch_\rho[f]) \geq \epsf$. Fattal and Ron~\cite{FR10} observed that $\dist(f, \switch_\rho[f]) \leq 2 \cdot \epsf$. (See \cite[Section~3.1]{KMS15} for a discussion of this.)

For an ordered set $T$ and a vector $\pi \in \{Y,N\}^{|T|}$, let $T \star \pi$ denote the ordered set consisting of only the elements in $T$ whose corresponding position in $\pi$ is $Y$. For example, if $T = (5, 1, 7, 4, 2, 9)$ and $\pi = (Y, Y, N, Y, N, N)$, then $T \star \pi = (5, 1, 4)$. For any finite set $S$, let $\mathcal{P}(S)$ denote the uniform distribution supported on the set of all permutations of $S$.

The following theorem is implicit in~\cite[Section~4]{KMS15}.
\begin{theorem}\label{thm:kms-iso}
For a function $f \colon \hypercube \to \{0,1\}$, there exists a constant $C > 0$ such that
\[
C \cdot \Ex_{\bx \sim \{0,1\}^\n} \left[ \sqrt{I_f^-(\bx)} \right] \geq
\Ex_{\boldsymbol{\lambda} \sim \mathcal{P}([\n])} \left[ \dist(f, \switch_{\boldsymbol{\lambda}} [f]) \right] - \Ex_{\substack{\boldsymbol{\rho} \sim \mathcal{P}([\n]) \\ \boldsymbol{\pi} \sim \{Y,N\}^{\n} }} \left[ \dist(f, \switch_{\boldsymbol{\rho} \star \boldsymbol{\pi}}[f]) \right].
\]
\end{theorem}

Consider the quantity
$\Ex_{\bw\sim\{0,1\}^{\bS}} \left[\sqrt{I_{f(\cdot, \bz)}^-(\bw)}\right]$
in the statement of Proposition~\ref{prop:stronger-bound}. From Theorem~\ref{thm:kms-iso}, we get
\[
C \cdot \Ex_{\bw \sim\{0,1\}^{\bS}} \left[\sqrt{I_{f(\cdot, \bz)}^-(\bw)}\right]
\geq
\Ex_{\boldsymbol{\lambda} \sim \mathcal{P}(\bS)} \left[ \dist(f(\cdot, \bz), \switch_{\boldsymbol{\lambda}} [f(\cdot, \bz)]) \right] \; - \! \Ex_{\substack{\boldsymbol{\rho} \sim \mathcal{P}(\bS) \\ \boldsymbol{\pi} \sim \{Y,N\}^{|\bS|} }} \left[ \dist(f(\cdot, \bz), \switch_{\boldsymbol{\rho} \star \boldsymbol{\pi}}[f(\cdot, \bz)]) \right].
\]
Taking the expectation over $\bS \sim \calS(p)$, we get
\begin{multline} \label{eq:exp-dist}
C \cdot \Ex_{\substack{\bS \sim \calS(p) \\ \bz \sim \{0,1\}^{\overline{\bS}}}} \left[ \Ex_{\bw \sim \{0,1\}^{\bS}} \left[\sqrt{I_{f(\cdot, \bz)}^-(\bw)} \right] \right]
\geq  \Ex_{\substack{\bS \sim \calS(p) \\ \bz \sim \{0,1\}^{\overline{\bS}}}} \left[ \Ex_{\boldsymbol{\lambda} \sim \mathcal{P}(\bS)} \left[ \dist(f(\cdot, \bz), \switch_{\boldsymbol{\lambda}} [f(\cdot, \bz)]) \right] \right] \\
 -
\Ex_{\substack{\bS \sim \calS(p) \\ \bz \sim \{0,1\}^{\overline{\bS}}}} \bigg[ \Ex_{\substack{\boldsymbol{\rho} \sim \mathcal{P}(\bS) \\ \boldsymbol{\pi} \sim \{Y,N\}^{|\bS|} }} \left[ \dist(f(\cdot, \bz), \switch_{\boldsymbol{\rho} \star \boldsymbol{\pi}}[f(\cdot, \bz)]) \right] \bigg].
\end{multline}
For ease of notation, let $\Psi_f(p)$ denote the first term on the right-hand side of (\ref{eq:exp-dist}):
\[
\Psi_f(p) = \Ex_{\substack{\bS \sim \calS(p) \\ \bz \sim \{0,1\}^{\overline{\bS}}}} \left[ \Ex_{\boldsymbol{\lambda} \sim \mathcal{P}(\bS)} \left[ \dist(f(\cdot, \bz), \switch_{\boldsymbol{\lambda}} [f(\cdot, \bz)]) \right] \right].
\]
We can also express the second term on the right-hand side of (\ref{eq:exp-dist}) using this notation:
\[
\Ex_{\substack{\bS \sim \calS(p) \\ \bz \sim \{0,1\}^{\overline{\bS}}}} \bigg[ \Ex_{\substack{\boldsymbol{\rho} \sim \mathcal{P}(\bS) \\ \boldsymbol{\pi} \sim \{Y,N\}^{|\bS|} }} \left[ \dist(f(\cdot, \bz), \switch_{\boldsymbol{\rho} \star \boldsymbol{\pi}}[f(\cdot, \bz)]) \right] \bigg] = \Psi_f(p/2),
\]
since the probability that an arbitrary $i \in [\n]$ appears in $\rho \star \pi$ is equal to the product of the probability that $i \in \bS$ and the probability that the position where $i$ appears in $\pi$ has $Y$. By \eqref{eq:exp-dist} and Proposition~\ref{prop:stronger-bound}, for every $p \in [0, 1]$, we get
\begin{align}
\Psi_f(p) - \Psi_f(p/2)
&\leq C \cdot \Ex_{\substack{\bS \sim \calS(p) \\ \bz \sim \{0,1\}^{\overline{\bS}}}} \left[ \Ex_{\bw \sim \{0,1\}^{\bS}} \left[\sqrt{I_{f(\cdot, \bz)}^-(\bw)} \right] \right] \notag \\
&\leq C \cdot \sqrt{p} \cdot \Ex_{\bx \sim \hypercube} \left[\sqrt{I_{f}^-(\bx)}\right].  \label{eq:p-dist}
\end{align}

It follows from the analysis of Dodis et al.~\cite{DGLRRS99} that $\Psi_f(1) \geq \eps_f$. Also note that $\Psi_f(0) = 0$.
We apply \eqref{eq:p-dist} for $p = 1, \frac{1}{2}, \frac{1}{4} \dots$ and consider the resulting telescoping sum:
\begin{align*}
\eps_f\leq \Psi_f(1) - \Psi_f(0) & = \sum_{i=0}^{\infty} \left(\Psi_f(2^{-i}) - \Psi_f(2^{-i-1})\right) \\
& \leq C \cdot \sum_{i=0}^{\infty} 2^{-i/2} \cdot \Ex_{\bx \sim \hypercube}\left[\sqrt{I_{f}^-(\bx)}\right] \\
& \leq \frac{\sqrt{2}}{\sqrt{2}-1}\ C \cdot \Ex_{\bx \sim \hypercube} \left[ \sqrt{I_f^-(\bx)}\right],
\end{align*}
completing the proof of Theorem~\ref{thm:kms-tal}. Similarly, Proposition~\ref{prop:improv-tal-obj} implies Theorem~\ref{thm:kms}.

\section{A Nonadaptive Lower Bound for Erasure-Resilient Testers}\label{sec:ER-lb}
In this section, we prove Theorem~\ref{thm:mono-ER-lb} that gives a lower bound on the query complexity of erasure-resilient testers of monotonicity, unateness and the $k$-junta property. We prove the lower bound by constructing two distributions $\Dyes$ and $\Dno$ on input functions $f\colon \hypercube\to \{0,1,\bot\}$ that are hard to distinguish for every nonadaptive tester and then applying Yao's Minimax principle~\cite{Yao77}.

Recall that $\bot$ denotes an erased function value. We say that a function $f\colon \hypercube \to \{0,1,\bot\}$ is {\em $\alpha$-erased} if at most an $\alpha$ fraction of its values are erased. If $\alpha$ is not specified, we call such a function {\em partially erased}. A {\em completion} of a partially erased function $f \colon \hypercube\to \{0,1,\bot\}$ is a function $f' \colon \hypercube\to \{0,1\}$  that agrees with $f$ on all nonerased values, that is, for all $x\in \hypercube$, if $f(x) \neq \bot$ then $f'(x)=f(x)$.
A partially erased function
has property $\cal P$ 
if there exists
a completion of $f$ that has $\cal P$. A partially erased function is $\eps$-far from
having property $\cal P$ 
if every completion of $f$ is $\eps$-far from
having property $\cal P$.

\begin{proof}[Proof of Theorem~\ref{thm:mono-ER-lb}]
We start by defining distributions $\Dyes$ and $\Dno$ on $\alpha$-erased functions. Later, we show that
$\Dyes$ is over monotone functions whereas most of the probability mass of $\Dno$ is on functions that are $\eps$-far from monotone. Interestingly, the same distributions work to prove our lower bounds for unateness and $k$-juntas: all functions in the support of $\Dyes$ are unate (because they are monotone) and also $n/2$-juntas whereas $\Dno$ is mostly supported on functions that are $\eps$-far from unate and $\eps$-far from $n/2$-juntas. The core of the argument is demonstrating that the two distributions are hard to distinguish for nonadaptive testers that make too few queries.

\begin{definition}
For a set $S\subseteq [n]$, let $\overline{S}$ denote the complement set  $[\n] \setminus S,$ and let $\{0,1\}^S$ denote the restriction of $\hypercube$ to the dimensions in $S$.
For every $x \in \hypercube$, let $|x|$ denote the Hamming weight of $x$, and let $x_S$ denote the vector $x \in \hypercube$ restricted to the dimensions in the set $S \subseteq [\n]$.
\end{definition}

Let $\n$ be a multiple of 4.
We first describe a collection of objects needed for defining the distributions $\Dyes$ and $\Dno$.
As before, we use the convention that random variables are boldface whereas fixed quantities use standard typeface.

\begin{itemize}
\item {\bf The set $M$ of {\em control dimensions} and the {\em control substring} $x_M$.} We partition $[n]$ into two sets of size $n/2$: the set $M$ of {\em control dimensions} and the set $\overline{M}$. We call $x_M$ the {\em control substring} of $x$.
    The random variable $\bM$ is a uniformly random subset of $[\n]$ of size $\n/2$.
\item {\bf The {\em subcube partition set} $P_{M}$ and {\em action} subcubes}. Let
$\Psi_{M} = \{x_M \in \{0,1\}^M : |x_M| =\frac{\n}{4}\}$
 denote the set of all control substrings of $x$ which lie in the middle layer of the subcube $\{0,1\}^M$.
 Each $z\in\Psi_M$ corresponds to a subcube of the form $\{0,1\}^{\overline{M}}$ with the vertex set comprised of points $x$ with $x_M=z.$ All such subcubes are called {\em action} subcubes.
 The {\em subcube partition set} $P_{M}$ is a subset of $\Psi_M$.
The random variable $\bP_M$ is a uniformly random subset of $\Psi_M$, that is, each $x_M \in \Psi_M$ is included in $\bP_M$ independently with probability $1/2$.
\item {\bf The functions $g_{M, P_M}$.} Recall the parameter $\kappa \in(0,1/2)$ from Theorem~\ref{thm:mono-ER-lb}. For a fixed setting of $M \subset [n]$ of size $n/2$ and a subcube partition set $P_M \subseteq \Psi_M$, we define the function $g_{M, P_{M}}\colon\hypercube \to {\{0,1,\bot,0^\star,1^\star\}}$ as follows:
$$
g_{M, P_M}(x) =
\begin{cases}
1 & \text{ if } |x_M| > \frac{\n}{4}; \\
0 & \text{ if } |x_M| < \frac{\n}{4}; \\
\bot & \text{ if } |x_M| = \tfrac{\n}{4} \text{ and } |x_{\overline{M}}| \in \left[\tfrac{\n}{4} - n^\kappa,  \tfrac{\n}{4} + n^\kappa\right];\\
1^\star & \text{ if } x_M \in P_M \text{ and } |x_{\overline{M}}| \notin \left[\tfrac{\n}{4} - n^\kappa,  \tfrac{\n}{4} + n^\kappa\right]; \\
0^\star & \text{ otherwise.}
\end{cases}
$$
\end{itemize}

To sample a function from $\Dyes$ and $\Dno$, we pick a uniformly random subset $\bM \subset [n]$ of size $\n/2$ of control dimensions and a uniformly random subcube partition set $\bP_{\bM}\subseteq\Psi_{\bM}$. A function sampled from $\Dyes$ and $\Dno$ is identical to $g_{\bM, \bP_{\bM}}$ on points $x \in \{0,1\}^n$ for which $g_{\bM, \bP_{\bM}}(x)\in \{0,1,\bot\}$, but differs on the remaining values (see Figures~\ref{fig:lb-construction-Dyes} and~\ref{fig:lb-construction-Dno}). Specifically, for functions in $\Dyes$, the values $0^\star$ and $1^\star$ are replaced with $0$ and $1$, respectively. That is, $\boldf^+ \sim \Dyes$ is defined by sampling $\bM$ and $\bP_{\bM}$ and letting
\[
\boldf^+(x) =
\begin{cases}
g_{\bM, \bP_{\bM}}(x) & \text{if } g_{\bM, \bP_{\bM}}(x)\in\{0,1,\bot\}; \\
0 & \text{if } g_{\bM, \bP_{\bM}}(x)=0^\star;\\
1 & \text{if } g_{\bM, \bP_{\bM}}(x)=1^\star. \\
\end{cases}
\]
For functions in $\Dno$, the value $0^\star$ is replaced with the majority function, denoted $\maj(\cdot),$ evaluated on the bits indexed by $\overline{\bM},$ whereas $1^\star$ is replaced with the antimajority of those bits.
That is, $\boldf^- \sim \Dno$ is defined by sampling $\bM$ and $\bP_{\bM}$ and letting
\[
\boldf^-(x) =
\begin{cases}
g_{\bM, \bP_{\bM}}(x) & \text{if } g_{\bM, \bP_{\bM}}(x)\in\{0,1,\bot\}; \\
\maj(x_{\obM}) & \text{if } g_{\bM, \bP_{\bM}}(x)=0^\star;\\
1 - \maj(x_{\obM}) & \text{if } g_{\bM,\bP_{\bM}}(x)=1^\star. \\
\end{cases}
\]

\begin{figure}
    \centering
    \includegraphics[scale=0.75]{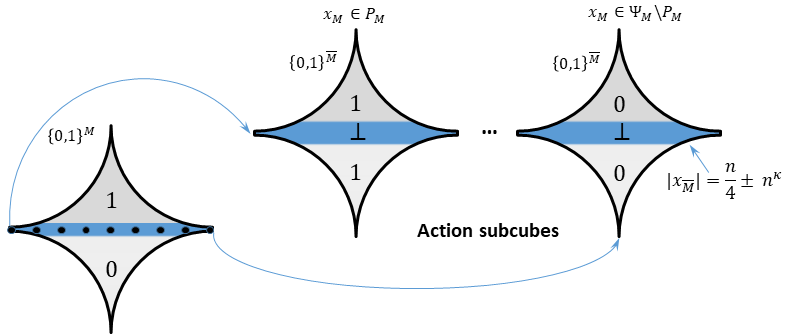}
    \caption{Functions $\boldf^+ \sim \Dyes$ defined with respect to control dimensions $M$ and the subcube partition set $P_{M}.$}
    \label{fig:lb-construction-Dyes}
\end{figure}

\begin{figure}
    \centering
    \includegraphics[scale=0.75]{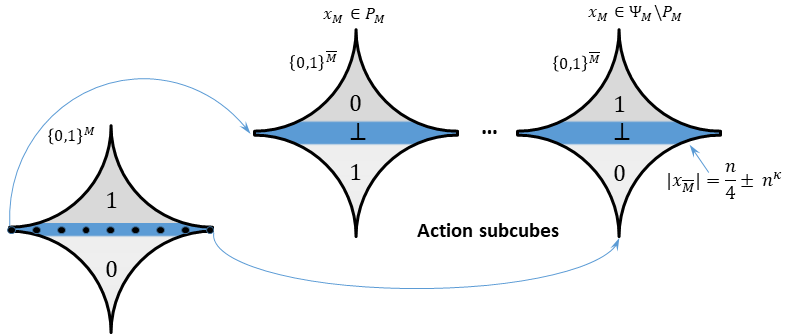}
    \caption{Functions $\boldf^- \sim \Dno$ defined with respect to control dimensions $M$ and the subcube partition set $P_{M}.$}
    \label{fig:lb-construction-Dno}
\end{figure}

\begin{lemma}\label{lem:alpha}
There is an $\alpha = O(1/n^{1-\kappa})$
for which every function in the support of the distributions $\Dyes$ and $\Dno$ is $\alpha$-erased.
\end{lemma}

\begin{proof}
Consider a partially erased function $f\colon \{0,1\}^n \to \{0,1\}$ in the support of $\Dyes$ or $\Dno$. Let $M$ be the set of control dimensions used in defining $f$. The function $f$ is erased in the middle $2n^{\kappa} + 1$ layers of every action subcube $\{0,1\}^{\overline{M}}$ for which the control substring $x_M$ is in the middle layer of the subcube $\{0,1\}^M$. 
Since $|M|=|\overline{M}|=n/2$, the number of erased points is at most
\begin{align*}
\binom{\n/2}{\n/4} \cdot \binom{\n/2}{\n/4} (2n^\kappa+1)
=O\left(\left(\frac{2^{n/2}}{\sqrt{n/2}}\right)^2 n^\kappa\right)
=O\left(\frac{2^n}{n^{1-\kappa}}\right).
\end{align*}
Thus, the fraction of erasures in the constructed functions is $O(1/n^{1-\kappa})$.
\end{proof}

\subsection{Distance to Monotonicity, Unateness, and Being a Junta}
In this section, we analyze the distance to monotonicity, unateness, and being an $n/2$-junta for functions in the support of our hard distributions, $\Dyes$ and $\Dno$.
\begin{lemma}
Every function $f^+$ in the support of $\Dyes$ is monotone, unate, and an $n/2$-junta.
\end{lemma}
\begin{proof}
Consider a partially erased function $f^+$ in the support of $\Dyes.$
We show that $f^+$ can be completed to a monotone $n/2$-junta.
Let $M$  be the set of control dimensions and $P_M$ be the subcube partition set used in defining $f^+$. Define a completion $f'\colon \hypercube \to \{0,1\}$ of $f^+$ as follows. For all $x \in \hypercube$ with $f^+(x) = \bot$,
\begin{align*}
f'(x) =
\begin{cases}
1 &\text{if } x_M \in P_M, \\
0 &\text{otherwise.}
\end{cases}
\end{align*}
Since $|M|=n/2$ and $f'$ depends only on coordinates in $M,$ it is an $\n/2$-junta.

To prove that $f'$ is monotone, we show it is the disjunction of monotone functions. Let $h:\{0,1\}^n\to\{0,1\}$ be the indicator for $|x_M|>n/4$, that is, $h(x)=1$ iff $|x_M|>n/4$. For every $y\in P_M$, let $h_y:\{0,1\}^n\to\{0,1\}$ be the indicator for $x_M= y_M$. Functions $h$ and $h_y$ for all $y\in P_M$ are monotone, and $f'$ is the disjunction of these functions.
Hence, $f^+$ can be completed to a monotone $n/2$-junta.  Since monotone functions are unate, Lemma~\ref{lem:eps} holds.
\end{proof}

Next we define an event called $\eventfar$, observe that it happens with high probability,  and show that functions distributed according to $\Dno$ conditioned on $\eventfar$ are far from the three properties we are considering.
\begin{definition}[Event $\eventfar$ and distribution $\Dhat$]
Let $\eventfar$ be the event that $$\frac{|\Psi_M|}3\leq |P_M|\leq \frac{2|\Psi_M|}3.$$
Define distribution $\Dhat=\Dno|_{\eventfar}$.
\end{definition}

Recall that each element of $\Psi_{M}$ is included in $\bP_{M}$ independently with probability 1/2.
Therefore, for sufficiently large $n$, by a Hoeffding bound,
\begin{align}\label{eq:not-far}
\Pr[\overline{\eventfar}]
\leq 2 e^{2|\Psi_M|\cdot \frac 1 {6^2}}
\leq \frac 1{30}
\end{align}
 over the draw of $\boldf^- \sim \Dno$.
\begin{lemma}\label{lem:eps}
Every function $f^-$ in the support of $\Dhat$ has distance at least $\eps = \Omega\left(\frac{1}{\sqrt{\n}}\right)$ from monotonicity, unateness, and being an $n/2$-junta.
\end{lemma}

\begin{proof}
Consider a partially erased function $f^-$ in the support of $\Dhat$, and  let $M$ and $P_M$ be the set of control dimensions and the subcube partition set  used in defining $f^-$, respectively. We start by proving that $f^-$ is far from monotone. Even though it is technically not necessary, because it will follow from the fact that $f^-$ is far from unate, we choose to do it for ease of presentation. After that, we prove Claim~\ref{claim:eps-junta} that states that $f^-$ is far from being an $n/2$-junta. Finally, we build on ideas in these two proofs to show in Claim~\ref{claim:eps-unate} that $f^-$ is far from unate.

\begin{definition}[Decreasing pair]\label{def:decreasing-pair}
A pair of points $\{x,y\}$ is {\em decreasing} with respect to a function  $f$ if $x\prec y$, but $f(x)=1$ and $f(y)=0.$
\end{definition}
First, we prove that $f^-$ is $\eps$-far from monotone for some $\eps=\Omega(1/\sqrt{n})$ by showing that there exists a large matching of decreasing pairs with respect to $f^-$.
Consider all action subcubes  $\{0,1\}^{\overline{M}}$ for which $x_M\in P_M$. By definition, each such subcube is $n/2$-dimensional, and hence contains $2^{n/2}$ points. Since $\eventfar$ holds, there are at least $\binom{n/2}{n/4}\cdot \frac 1 3=\Omega(\frac {2^{n/2}} {\sqrt{n}})$ such subcubes.
By standard arguments (see\footnote{Fischer et al.~\cite[Lemma 22]{FLNRRS02} prove that {\em trimmed anti-oligarchy} functions are far from monotone by applying Hall's Theorem to argue that there is a large matching of decreasing edges. Given $B\subseteq[n],$ the {\em trimmed anti-oligarchy} function $f_B(x)$ is 1 if the Hamming weight of $x$ is large, 0 if it is small, and is equal to $1-Maj(x_B)$, otherwise. In particular, when $B=[n],$ this function is equal to $1-Maj(x)$ on a constant fraction of points $x$ (located in the middle layers of the hypercube). The functions on the actions subcubes under consideration in our lower are also anti-majorities on a constant fraction of all points $x$ (located in all but the middle layers of the subcubes).}, e.g.,~\cite[Lemma 22]{FLNRRS02}), in each such action subcube,
there is a matching of size $\Omega(2^{n/2})$ consisting of decreasing pairs with respect to $f^-$.
The values of at least half of the points in the matching need to be changed to make the function  $f^-$ monotone. Moreover, at least two thirds of the points in the action subcube are nonerased, and this matching contains all of them.
Since $\Omega(1/\sqrt{n})$ fraction of points participates in action subcubes with  $x_M\in P_M$, and the values of at least a third of these points need to be changed to make $f^-$ monotone, the distance from $f^-$ to monotonicity is $\Omega(1/\sqrt{n})$.

Next, recall from Definition~\ref{def:hypercube} that, for $i \in [\n]$, a hypercube edge whose endpoints differ in the $\ord{i}$ coordinate is called an $i$-edge.
\begin{definition}[A constant edge, independence of a variable]
An edge $\{x,y\}$ is {\em constant} with respect to a function $f$ if $f(x)=f(y)$. An edge that is neither decreasing nor constant is called {\em increasing}. A function $f$ is {\em independent} of a variable $i \in [n]$ if all $i$-edges are constant with respect to $f$.
\end{definition}

\begin{claim}\label{claim:eps-junta}
Every function $f^-$ in the support of $\Dhat$ has distance at least $\eps = \Omega\left(\frac{1}{\sqrt{\n}}\right)$ from being an $n/2$-junta.
\end{claim}

\begin{proof}
Observe that $f^-$ is an $n/2$-junta iff it is independent of at least  $n/2$ variables.

Next, for each $i\in M$, we show that $f^-$ is $\eps$-far from being independent of $i$. Fix $i\in M$.
We construct a large set $\match{i}$ of increasing 
$i$-edges $\{x,y\}$.
At least one of $f^-(x)$ and $f^-(y)$ for each such edge has to change to make $f^-$ independent of $i$. Since $\match{i}$ is a matching, $|\match{i}|/2^n$ is a lower bound on the distance from $f^-$ to functions that do not depend on variable $i$.

Recall that the set $\Psi_M = \{ x_M \in \{0,1\}^M : {|x_M|=\frac{\n}{4}} \}$ is the set of all control substrings of $x$ that lie in the middle layer of the subcube $\{0,1\}^M$. We define, for every dimension $i\in M,$
\[
\match{i}=\{\{x,y\} : \{x,y\} \text{ is an $i$-edge, } x_M \in \Psi_M, \text{ and } f(x)=x_i\}.
\]
Note that $x_i\neq y_i$ and, by construction of functions $g_{M, P_M}$,
we have $f^-(y)=g_{M, P_M}(y)=y_i.$ Therefore, $f^-(x)\neq f^-(y)$ for all $i$-edges $\{x, y\} \in \match{i}$.
(Note that all edges in $\match{i}$ are increasing. This fact will be used in the proof of Claim~\ref{claim:eps-unate}.)
For each $x_M\in\Psi_M$, more than a third of the points $x$ in the corresponding action subcube are assigned $f^-(x)=0$, and the same holds for $f^-(x)=1$. Since each action subcube has $2^{\n/2}$ points, the size of $\match{i}$ is at least
$\frac{1}{3} \cdot 2^{\n/2} \cdot |\Psi_M| = \frac{1}{3} \cdot 2^{\n/2} \cdot \binom{\n/2}{\n/4} = \Omega(\frac{2^n}{\sqrt{\n}})$. That is, the distance from $f^-$ to being independent of variable $i$ is at least $\eps$, where $\eps=\Omega(\frac 1 {\sqrt{n}})$.

Thus, if we change less than an $\eps$ fraction of values of $f^-$, we cannot eliminate the dependence on any of the $\n/2$ variables in $M$. The only remaining possibility to make $f^-$ an $n/2$-junta with fewer than $\eps\cdot 2^n$ modifications is to eliminate the dependence on all variables in $\overline{M}$. This can happen only if the modified function becomes constant on all the action subcubes, which again requires changing at least $\frac{1}{3} \cdot 2^{\n/2} \cdot |\Psi_M|$ values of $f^-$. Thus, $f^-$ is $\eps$-far from the set of $n/2$-juntas, where $\eps=\Omega(\frac 1 {\sqrt{n}})$.
\end{proof}

\begin{claim}\label{claim:eps-unate}
Every function $f^{-}$ in the support of $\hat{\mathcal{D}}^-$ has distance at least $\eps = \Omega\left(\frac{1}{\sqrt{n}} \right)$ from unateness.
\end{claim}

\begin{proof}
To show that $f^{-}$ is $\eps$-far from unate, let $r \in \{0,1\}^n$ be an arbitrary assignment of directions for each variable in $[n]$ (that is, $r_i = 0$ signifies that the variable $i$ is monotone non-decreasing, and $r_i = 1$ signifies that the variable $i$ is monotone non-increasing). We use $\oplus$ to represent bit-wise XOR of two vectors. It suffices to show that the function $g \colon \{0,1\}^n \to \{0, 1\},$ given by $g(x) = f^{-}(x \oplus r),$ is $\eps$-far from monotone.

Notice that if $r_i = 1$ for some $i \in M$, the function $g$ is $\Omega(1/\sqrt{n})$-far from monotone.
To see this, recall that we constructed a matching $\match{i}$ of $\Omega(2^n/\sqrt{n})$ increasing
$i$-edges in $f^-$ in the proof of Claim~\ref{claim:eps-junta}. If $r_i =1$, these $i$-edges are decreasing with respect to $g$, so the function is $\Omega(1/\sqrt{n})$-far from monotone. Henceforth, we assume $r_i = 0$ for all $i \in M$.

Let $R$ be the subset of dimensions $\{i\in \overline{M} : r_i = 1\}$. First consider the case when $|R|\leq n/4$.
Recall that each action subcube contains $2^{n/2}$ points. Consider all action subcubes  $\{0,1\}^{\overline{M}}$ for which $x_M\in P_M$.  Since $\eventfar$ holds, there are  $\Omega(\frac {2^{n/2}} {\sqrt{n}})$ such subcubes, as we argued before. We show that each of them contains a matching of $\Omega(2^{n/2})$ decreasing pairs with respect to $g,$ thus proving that $g$ is $\Omega(1/\sqrt{n})$-far from monotone.

We further partition action subcubes into smaller subcubes, each of which contains all points $x\in\hypercube$ with the same $x_R$ (and the same control substring $x_M$). Let $\cal C$ be the set of all such subcubes with $x_M\in P_M$ and the Hamming weight $|x_R|$ in the range $\big[\frac{|R|}2-\sqrt{n},\frac{|R|}2+\sqrt{n}\big].$ Since a constant fraction of $z\in\{0,1\}^R$ is in the specified range,
$\cal C$ contains $\Omega(2^{|R|})$ smaller subcubes for each action subcube we are considering.
It remains to show that each subcube in $\cal C$ contains a matching of $\Omega(2^{n/2-|R|})$ decreasing pairs with respect to $g.$

By definition of $f^-$, the value of $f^-(y)$ depends only on the Hamming weight of $y_{\overline M}$ for all points $y$ with $y_M\in P_M.$ In particular, $f^-(y)=0$ if $|y_{\overline M}|\geq \frac n 4 + n^\kappa$ and $f^-(y)=1$ if $|y_{\overline M}|\leq \frac n 4 - n^\kappa.$
That is, for every $x$ with $x_M\in P_M$,
\[
g(x)=f^-(x\oplus r)
=\begin{cases}
0 \text{ if } |(x\oplus r)_{\overline M}|\geq \frac n 4 + n^\kappa;\\
1 \text{ if } |(x\oplus r)_{\overline M}|\leq \frac n 4 - n^\kappa.
\end{cases}
\]
Recall that $R\subseteq {\overline M}.$
Since $(x\oplus r)_i=1-x_i$ for all $i\in R$, and $(x\oplus r)_i=x_i$ for all $i\notin R$, we get
\[
|(x\oplus r)_{\overline M}|= |(x\oplus r)_{R}|+ |(x\oplus r)_{{\overline M}\setminus R}|
= |R|-|x_R|+ |x_{{\overline M}\setminus R}|.
\]

Fix a subcube in $\cal C$ and note that it is of the form $\{0,1\}^{{\overline M}\setminus R}.$ Recall that $|\overline M|=n/2$, so the subcube contains $2^{n/2-|R|}$ points. The quantity $|R|-|x_R|$ is the same for all points $x$ in the subcube. Moreover, this quantity is in the range   $\big[\frac{|R|}2-\sqrt{n},\frac{|R|}2+\sqrt{n}\big].$ We claim that $g$ evaluates to 1 on the points in the bottom layers of the subcube, and that it evaluates to 0 on the points in the top layers.
Specifically, if $|x_{{\overline M}\setminus R}|\leq \frac {|{\overline M}\setminus R|}2-\sqrt{n}-n^\kappa,$ which holds for a constant fraction of points $x$ in the subcube, then
\[
|(x\oplus r)_{\overline M}|
= |R|-|x_R|+ |x_{{\overline M}\setminus R}|
\leq \frac{|R|}2+\sqrt{n} + \frac {|{\overline M}\setminus R|}2-\sqrt{n}-n^\kappa
= \frac {|{\overline M}|}2-n^\kappa
=\frac n 4 -n^\kappa.
\]
That is, in this case, $g(x)=0.$
Similarly, if $|y_{{\overline M}\setminus R}|\geq \frac {|{\overline M}\setminus R|}2+\sqrt{n}+n^\kappa,$ which also holds for a constant fraction of points in the subcube, then
\[
|(y\oplus r)_{\overline M}|
\geq \frac{|R|}2-\sqrt{n} + \frac {|{\overline M}\setminus R|}2+\sqrt{n}+n^\kappa
=\frac n 4 +n^\kappa.
\]
That is, in this case, $g(y)=0.$  Similarly to the case of monotonicity, by a standard argument \cite[Lemma 22]{FLNRRS02}, there is a matching of decreasing pairs $\{x,y\}$ with respect to $g$ that matches all the points $x$ and $y$ described above
and, consequently has size $\Omega(2^{n - |R|}/\sqrt{n})$.
We obtained the desired bound on the matching size and, therefore, on the distance to unateness.

The argument for the case when $|R| > n/4$ follows similarly by symmetry. We consider action subcubes for which the control substrings are not in $P_M$. For each such subcube, $f^-$ evaluates to a majority on $x_{\overline M}$, and variables in $R$ are flipped in $g$, allowing us to construct a large matching of decreasing pairs with respect to $g$.
\end{proof}

This completes the proof of Lemma~\ref{lem:eps}.
\end{proof}

\subsection{Indistinguishability of the Hard Distributions}
Next, we show that the distributions $\Dyes$ and $\Dhat$ are hard to distinguish for nonadaptive testers.
Consider a deterministic tester that makes $q$ queries. Let $\bA_q(f)$ be the sequence of query-answer pairs $(x,f(x))$ obtained by the tester on input $f$. For every distribution $\cal D$ on input functions, define ${\cal D}$-view to be the distribution on $\bA(\boldf^+)$ when $\boldf^+\sim{\cal D}$. We use the version of Yao's principle stated in~\cite{RS06} that asserts that to prove a lower bound $q$ on the worst-case query complexity of a randomized algorithm, it is enough to give two distributions $\Dyes$ and $\Dhat,$ on positive and negative instances, respectively, for which
the statistical distance between $\Dyes$-view and $\Dhat$-view is less than 1/3.

We start by analyzing how a tester can distinguish $\Dyes$ and $\Dno$. The key point is the only way to do it is by querying a pair of points $x,y \in \hypercube$ that fall in the same action subcube, but in different nonerased layers---one below the erased layers, the other above the erased layers. If the tester queries no such pair, then its view (that is, the distribution on the sequence of query-answer pairs it obtains) is identical for the two cases: $\boldf^+\sim\Dyes$ and $\boldf^-\sim\Dno$.

Let $\eventbad$ be the event that one of the $\binom{q}{2}$ pairs of points the tester queries ends up in the same action subcube, on different sides of erasures, as discussed above. Next, we show that conditioned on $\eventbad$ not occurring, the view of the tester is the same for both distributions.
\begin{claim}\label{claim:distributions}
$\displaystyle
\Dyes\text{-view}|_{\overline{\eventbad}} = \Dno\text{-view}|_{\overline{\eventbad}}.
$
\end{claim}
\begin{proof}
Consider an arbitrary fixed set of control dimensions $M \subset [n]$ of size $n/2$, and let $\Psi_{M}$ be the set of control substrings. Partition the set of queries $x_1, \dots, x_q \in \{0,1\}^n$ with nonerased values according to their control substrings in $\Psi_{M}$; namely, for $z \in \Psi_M$, let $Q_z \subset \{ x_1, \dots, x_q\}$ be the set of queries $x$ for which $x_M=z$  and $x_{\overline{M}}$ does not lie in the middle $2n^{\kappa}$ layers of its action subcube. This partition depends only on $M$, and not on the set $\bP_M$.

Suppose, furthermore, that any two queries $x, y \in \hypercube$ falling in the same action subcube (i.e., from the same part $Q_z$) are either both above the erased layers or both below the erased layers, and notice that this implies $\boldf(x) = \boldf(y)$ whenever $x, y \in Q_z$, for $\boldf$ sampled from $\Dyes$ as well as $\Dno$. We will show that the distribution over answers to queries for a fixed $M$ (and a random $\bP_M)$ is exactly the same for $\Dyes$ and $\Dno$. In both cases, the answer to a query $x\in \hypercube$ is a function of the control substring $x_{M}$ and is independent for two queries $x, y \in \hypercube$ with different control substrings. In fact, for each control substring $z \in \Psi_{M},$ the value of every $\boldf(x)$ with $x_M = z$ is a uniformly random bit.

In order to see why, notice that for $\boldf^+ \sim \Dyes$, every $x \in Q_z$ has value $\boldf^+(x) = 1$ if and only if $z \in \bP_{M}$, which occurs with probability $1/2$. For $\boldf^- \sim \Dno$, we have two cases. The first case is when $x \in Q_z$ has $|x_{\overline{M}}| \leq n/4 - n^{\kappa}$; then, $\boldf^-(x) = 1$ if and only if $z \in \bP_M$, which occurs with probability $1/2$. The second case is when $x\in Q_z$ has $|x_{\overline{M}}| \geq n/4 + n^{\kappa}$; then, $\boldf^+(x) = 0$ if and only if $z \notin \bP_M$, which occurs with probability $1/2$.
\end{proof}

\begin{claim}\label{claim:bad}
For a deterministic nonadaptive tester making $q \leq 2^{n^\kappa}$ queries,
 $$\Pr[\eventbad]<1/8$$
over the draw of $\bM \subset [n]$ of size $n/2.$
\end{claim}
\begin{proof}
We want to bound the probability over the draw of $\bM$ that any two queries $x, y \in \{0,1\}^n$ which have the same control substring are either both above the erased layers, or both below the erased layers in their action subcube. In particular, observe that the Hamming weights of any such $x$ and $y$ must differ by at least
$2n^\kappa+2.$
Consequently, $x$ and $y$ differ on at least
$2n^\kappa+2$
bits. Next, we upper bound the probability that two queries $x, y \in \{0,1\}^n$ which differ on at least $2n^{\kappa} + 2$ dimensions have the same setting of coordinates in $\bM$.

Let $T = \{i \in [n] : x_i \neq y_i \}$ denote the set of all coordinates on which the points $x$ and $y$ differ. Then $|T|\geq 2n^\kappa+2$. Observe that $x_{\bM} = y_{\bM}$ iff $T \cap \bM = \emptyset$. Since $\bM$ is a uniformly random subset of $[n]$ of size $n/2$,
\begin{align*}
\Pr_{\bM}[T \cap \bM = \emptyset] = \frac{\binom{n-|T|}{n/2}}{\binom{n}{n/2}} &= \frac{\frac{(n-|T|)!}{(n/2)!\cdot(n/2-|T|)!}}{\frac{n!}{(n/2)!\cdot(n/2)!}} \\
&= \frac{(n/2)!}{(n/2-|T|)!} \cdot \frac{(n-|T|)!}{n!} \\
&= \frac{n/2\cdot(n/2-1)\cdots(n/2-|T|+1)}{n\cdot(n-1)\cdots(n-|T|+1)}
\leq 2^{-|T|}.
\end{align*}
Let $\eventbad$ be the event that one of the $\binom{q}{2}$ pairs of points the tester queries ends up in the same action subcube, on different sides of erasures, as discussed above. Then, by a union bound,
\[\Pr[\eventbad] < \frac{q^2}2\cdot 2^{-|T|} \leq \frac 1 2 \cdot 2^{2n^\kappa}\cdot 2^{-2n^\kappa-2}= \frac 1 8.\]

\end{proof}
\begin{definition}[Notation for statistical distance]
For two distributions ${\cal D}_1$ and ${\cal D}_2$ and a constant $\delta,$ let ${\cal D}_1\approx_\delta {\cal D}_2$ denote that the statistical distance between ${\cal D}_1$ and ${\cal D}_2$ is at most $\delta.$
\end{definition}

\begin{lemma}\label{lem:indistinguishability}
For a deterministic nonadaptive tester making $q \leq 2^{n^\kappa}$ queries,
\[\Dyes\text{-view} \approx_{9/28}\Dhat\text{-view}.\]
\end{lemma}
\begin{proof}
By Claim~\ref{claim:distributions}, conditioned on $\eventbad$ not occurring, the view of the tester is the same for distributions $\Dyes$ and $\Dno$:
\[
\Dyes\text{-view}|_{\overline{\eventbad}} =\Dno\text{-view}|_{\overline{\eventbad}}.
\]
Conditioning on $\overline{\eventbad}$ does not significantly change the view distributions. We use the following claim \cite[Claim 4]{RS06} to formalize this statement.

\begin{claim}[\cite{RS06}]\label{claim:conditioning} Let $E$ be an event that happens with probability at least $\delta=1-1/a$ under the distribution
${\cal D}$ and let ${\cal B}$ denote distribution ${\cal D}|_E$. Then
${\cal B}\approx_{\delta'} {\cal D}$ where $\delta'=1/(a-1)$.
\end{claim}

Applying  Claims~\ref{claim:bad} and~\ref{claim:conditioning} twice, we get
\[
\Dyes\text{-view}\approx_{1/7}
\Dyes\text{-view}|_{\overline{\eventbad}}
=\Dno\text{-view}|_{\overline{\eventbad}}
\approx_{1/7}\Dno\text{-view}.
\]
Similarly, recalling that $\Dhat=\Dno|_{\eventfar},$ by (\ref{eq:not-far}) and Claim~\ref{claim:conditioning}, we get
$$\Dno\text{-view} \approx_{1/29} \Dno\text{-view}|_{\eventfar}
=\Dhat\text{-view}.$$
Since $1/7+1/7+1/29<9/28$, this completes the proof of Lemma~\ref{lem:indistinguishability}.
\end{proof}

Theorem~\ref{thm:mono-ER-lb} follows by Yao's Principle.
\end{proof}

\section{From Tolerant Testing to Distance Approximation}\label{sec:tolerant-to-dist-approx}
The following theorem shows how to convert a tolerant tester to a distance approximation algorithm.

\begin{theorem} \label{thm:tolerant-to-dist-approx}
Consider input objects whose size is measured with respect to a parameter $n$, and let $\cP$ be a property of these objects.
Let $g(n)$ be a polynomial in $n$ such that $g(n) = \omega(1)$.
Let $\cA$ be a tolerant tester of the property $\cP$ that gets a parameter $\eps \in (0, 1/2)$ and oracle access to an input object $f$, makes $Q_\cA (n, \eps)$ queries, and outputs $\close$ or $\far$ as follows:
\begin{enumerate}
\item If $\dist(f, \cP) \leq \frac{\eps}{g(n)}$, then $\cA(\eps, f)$ outputs $\close$ with probability at least $2/3$;
\item If $\dist(f, \cP) \geq \eps$, then $\cA(\eps, f)$ outputs $\far$ with probability at least $2/3$.
\end{enumerate}
Suppose $Q_\cA (n, \eps)$ has at least linear dependence on $1/\eps$. Then, there exists a distance approximation algorithm $\mathcal{B}$ for the property $\cP$ that, given a parameter $\alpha \in (0,1/2)$ and oracle access to an input object $f$ such that $\dist(f, \mathcal{P}) \geq \alpha$, returns an estimate that, with probability at least $2/3$, is a $2g(n)$-approximation to $\dist(f, \cP)$. The query complexity of $\cB$ is $O(Q_\cA(n, \alpha) \cdot \log \log (1/\alpha))$.
\end{theorem}

\begin{proof}
The distance approximation algorithm $\cB$ is described in Algorithm~\ref{alg:tester-to-approx}.

\begin{algorithm}
\caption{Distance Approximation Algorithm $\mathcal{B}(\alpha, f)$} \label{alg:tester-to-approx}
\SetKwInOut{Input}{input} \SetKwInOut{Output}{output}
\SetKwFor{RepeatTimes}{repeat}{times}{end}

\Input{minimum distance parameter $\alpha \in (0,1/2)$; oracle access to the input object $f$.}
\Output{$2g(n)$ approximation $\hat{\eps}$ to $\dist(f, \mathcal{P})$.}

\DontPrintSemicolon
\BlankLine

\nl \label{step:b-for}\For{$i = 1$ {\em to} $\floor{\log \frac{1}{\alpha}}$}{
\nl Set $c \gets 0$ and $t \gets \ceil{18 \ln (3 \log(1/\alpha))}$. \;
\nl \RepeatTimes{$t$}{
\nl Run $\cA(2^{-i}, f)$. If it outputs \far, $c \gets c + 1$.\label{step:run-a}
}
\nl \lIf{$c \geq t/2$}{\Return{$\hat{\eps} = 2^{-i + 1}$}\label{step:amplified-a}.}
}
\nl \Return{$\hat{\eps} = \alpha$}.
\end{algorithm}

For each $i \in \left[\floor{\log\frac{1}{\alpha}}\right]$, let $E_i$ denote the event that $c$ calculated by the algorithm in the $\ord{i}$ iteration of the {\bf for} loop satisfies $c \geq t/2$ when $\dist(f, \cP) \geq 2^{-i}/g(n)$ and $c < t/2$ when $\dist(f, \cP) \leq 2^{-i}$.
Event $\overline{E_i}$ occurs when $\cA$ errs in at least $t/2$ of its runs in Step~\ref{step:run-a}.
In each run, $\cA$ errs with probability at most $1/3$.
By Hoeffding bound, for each $i \in \left[\floor{\log({1}/{\alpha})} \right]$, the probability $\Pr[\overline{E_i}]  < 1/(3 \log(1/\alpha))$.
Let $E = \bigcap_{i \in [\floor{\log (1/\alpha)} ]} E_i$. By the union bound, $\Pr[\overline{E}] < 1/3$, which implies that $\Pr[E] > {2}/{3}$.
For the rest of the proof, condition on the event $E$ happening.
If $\hat{\eps} = 2^{-i+1}$, then it implies that for all $j < i$,
the algorithm determined that $\dist(f, \cP) \leq 2^{-j}$. Substituting $j = i-1$ yields $\dist(f, \cP) \leq 2^{-i+1} = \hat{\eps}$. Similarly, in the $\ord{i}$ iteration, the algorithm determined that $\dist(f, \cP) \geq 2^{-i}/g(n) = \hat{\eps}/(2g(n))$.
Hence, with probability at least $2/3$, the output $\hat{\eps}$ of the algorithm $\cB$ satisfies $\dist(f, \cP) \leq \hat{\eps} \leq 2g(n) \cdot \dist(f, \cP)$, completing the proof of the approximation guarantee.

The query complexity of $\cB$ is at most $t \cdot \sum_{i=1}^{\floor{\log (1/\alpha)}} Q_\cA(n, 2^{-i}).$
Since $Q_\cA(n, \eps)$ has at least linear dependence on $1/\eps$, the sum
is dominated by the last term. Therefore, the query complexity is
 $O(Q_\cA(n, \alpha) \cdot
\log \log (1/\alpha))$.
\end{proof}

\paragraph{Acknowledgments.} We thank Deeparnab Chakrabarty and C. Seshadhri for useful discussions and, in particular, for mentioning an adaptive algorithm for approximating the distance to monotonicity up to a factor of $O(\sqrt{n})$.

\bibliography{references}

\end{document}